\documentclass[12pt,titlepage]{article}
\usepackage{amsmath,amsthm,amssymb}
\usepackage[dvips]{graphicx}
\usepackage{array}
\usepackage[mathscr]{eucal}
\usepackage{eurosym}
\usepackage{subfigure}
\usepackage{color}
\usepackage{listings}
\usepackage{multirow} 
\usepackage[ruled,vlined]{algorithm2e}
\usepackage{verbatim}
\usepackage{bbm}

\usepackage{tikz}
\usepackage{hyperref}

\usepackage{natbib}
\newtheorem{Theorem}{Theorem}[section]
\newtheorem{condition}{Condition}
\newtheorem*{remark}{Claim}

\newtheorem{proposition}[Theorem]{Proposition}
\newtheorem{theorem}{Theorem}[section]

\newtheorem{assumption}{Assumption}[section]
\theoremstyle{definition}
\newtheorem{definition}[Theorem]{Definition}

\theoremstyle{remark}
\numberwithin{equation}{section}
\numberwithin{figure}{section}
\numberwithin{table}{section}

\newcommand{\bi}{ \begin{itemize}  }
\newcommand{\ei}{\end{itemize}}

\newcommand{\bfor}{ \begin{eqnarray*} }
\newcommand{\efor}{\end{eqnarray*}}

\begin{document}

\title{Nonparametric Value-at-Risk via Sieve Estimation}
\author{Philipp Ratz$^{a}$  \\\\
        \small $^{a}$ Université du Québec à Montréal (UQAM), Montréal (Québec), Canada \\\\
        \small \tt{ratz.philipp@courrier.uqam.ca}
}

\maketitle

\begin{abstract}
	Artificial Neural Networks (ANN) have been employed for a range of modelling and prediction tasks using financial data. However, evidence on their predictive performance, especially for time-series data, has been mixed. Whereas some applications find that ANNs provide better forecasts than more traditional estimation techniques, others find that they barely outperform basic benchmarks. The present article aims to provide guidance as to when the use of ANNs might result in better results in a general setting. We propose a flexible nonparametric model and extend existing theoretical results for the rate of convergence to include the popular Rectified Linear Unit (ReLU) activation function and compare the rate to other nonparametric estimators. Finite sample properties are then studied with the help of Monte-Carlo simulations to provide further guidance. An application to estimate the Value-at-Risk of portfolios of varying sizes is also considered to show the practical implications. 
	\\[15mm]
	{\bf Keywords:} Sieve Estimation, Neural Networks, High-Dimensional Estimation, Value-at-Risk, Nonparametric Estimation, Semiparametric Estimation
		\\[3mm]
		\\[3mm]
\noindent{\bf Acknowledgement and Notes}: The author would like to thank Arthur Charpentier for his help and constructive comments. \\
All data and code can be found on \href{https://github.com/phi-ra/nonparametric_var}{https://github.com/phi-ra/nonparametric\_var}
\end{abstract}

\section{Introduction}

Often, economic or financial theory does not suggest a specific functional form on, for example, time series models. For the purpose of simplicity a (parametric) linear form is still the go-to standard in many applications. This seems contrary to the fact that the empirical evidence suggests that many economic time series seem to have nonlinear forms (see eg. \cite{HFD2006} for examples of financial phenomena). Nonlinear extensions of simple parametric models exist, but they often suffer from misspecification, as a functional form must still be imposed. In recent years, due to the increased availability of data and computing resources, nonparametric approaches have gained popularity, see for example \cite{chen2005nonparametric}, \cite{chen2008nonparametric} or \cite{scaillet2005nonparametric} for nonparametric estimation of risk indicators	or \cite{tzeremes2018financial} for an example of economic time series. Their model-free nature provides the possibility to estimate and forecast without many assumptions as in standard parametric models. This does not come without cost though, nonparametric methods suffer from the curse of dimensionality\footnote{Roughly, this means that nonparametric estimators tend to converge very slowly in higher dimensional problems}, which makes them inaccurate in high exactly the type of problems that are often encountered in financial applications.

%

Artificial Neural Networks (ANNs or simply Neural Networks) allow for nonlinear models without explicitly specifying the functional form and have become increasingly popular in recent years. Practical applications of ANNs in finance are by no means new though, the earliest surveys were already conducted in the 1990s (see eg. \cite{wong1998neural}). Whereas earlier models were rather simple, recent applications became increasingly complex and applicable to a wider range of tasks. For example, earlier approaches such as \cite{swanson1995model} used simple, single-hidden layer networks later on, eg. \cite{zhanggp2003} combined a parametric model and an ANN for its residuals to get a more complex model that provided more accurate forecasts than each model by itself. More recently, \cite{deep_mortgage_risk} explored increasingly complex networks in both terms of "depth" and combination of multiple sub-models to build a classifier for mortgage risk. Applications to specific tasks in Finance can for example be found in \cite{xu2016quantile} who modelled the Value-at-Risk (VaR) using an ANN or \cite{bucci2020realized} who forecasted the realized volatility and compared the predictive performance of traditional econometric methods with a range of different ANNs. Most of the applications find, that allowing for nonlinearities indeed increases the predictive performance when compared to standard linear models. 

The results from these articles stand somewhat in contrast to what was found in a large practical application summarised in \cite{makridakis2020m4} where submissions from the \emph{M4} forecasting competition were considered. There, pure Machine-Learning based nonparametric models were found to perform worse than some even some simple benchmarks. In their article, they argue that more should be done to understand why this might be the case. Interestingly though, in the \emph{M4} competition, combinations of estimators were found to perform best, in line with the findings of for example \cite{zhanggp2003} and \cite{deep_mortgage_risk}. Further, one of the main findings of \emph{M4} was that using information from other, potentially related, time-series substantially improved the forecasting accuracy when using nonparametric ML models. This on the other hand would seem counter-intuitive when considering that nonparametric techniques suffer from the curse of dimensionality\footnote{The best-performing model was indeed a form of hybrid ANN}. This serves as motivation to explore possible reasons for the apparent contradictions and propose a general framework that combines both a parametric statistical model with an ANN into a semiparametric model.


To investigate the statistical properties of the model, this article aims to dive slightly deeper into the theoretical foundations of ANNs than most applications and provides a convergence order for a semiparametric "hybrid" model. The results extend the findings of \cite{chen_racine_2001} to include ANNs that are modelled with the so-called ReLU activation functions. As will be shown, the use of the ReLU activation function has desirable properties when working with deeper networks, which was also shown in eg. \cite{deep_mortgage_risk} to lead to better predictive performance\footnote{A simple example as to why this might be the case, is given in the code repository}. The remainder of the article is organized as follows: Section \ref{sec:formal} formalizes the neural network as a sieve estimator which is then used to derive the convergence order. Section \ref{sec:finite} extends the theoretical results and studies the convergence order on finite samples, Section \ref{sec:portfolio} demonstrates the use of the results for forecasting value-at-risk during stress periods that often entail significant deviations from a linear form. Section \ref{sec:conc} then concludes. 

%
%
%
%
%
%
%
%
%
%
%
%
%
%
%

\section{Formalization of the Estimation}\label{sec:formal}

ANNs are often treated as a separate model class, but they bear close similarities to other known econometric estimators, \cite{kuan_white_1994} provide an excellent comparison of ANNs and more well-known econometric models. ANNs themselves are a diverse family of estimators that are linked by a common concept. In what follows, we will focus on the simplest form of the ANN, the feedforward, fully connected, multilayer perceptron (MLP). This has several reasons: keeping the model simple permits the usage of existing theoretical convergence results and heuristically, even if different ANN Architectures lead to better results than those of the MLP, the general findings should still be applicable. To establish convergence results, this section first defines ANNs as sieve estimators, which essentially follows \cite{handbook_econometrics_chen}. Next, in the spirit of the findings of \cite{makridakis2020m4}, a semiparametric model is proposed and the convergence results for both the nonparametric and the parametric part are presented and compared to other nonparametric estimators. 

\subsection{The basic form of the MLP}

We consider the simplest form of the MLP, where only one "hidden layer" is present in the model. To start with, consider a model that takes the inputs $x_i, i=1,...,p$ and maps them to some output units $j, j=1,...,\nu$. Each input is then weighted a factor $\gamma_{ij} \in \mathbb{R}$. In the output unit $j$ the weighted inputs are then combined by a simple additive rule, so each unit produces the following output:

$$ \sum_{i=1}^{p} x_i\gamma_{ij}, \;\; j = 1,...,\nu $$

For simplicity, we define  $\tilde{x}$ as the $x$-vector with an additional "bias unit" $x_0=1$, which corresponds to the intercept in a linear equation. Rewriting the above we get the general form of the input combination:

\begin{equation}\label{eq:linear_perceptron}
	f_j(x,\gamma) \equiv \tilde{x}^\top\gamma_j, \;\; j=1,...,\nu
\end{equation}

Note that \eqref{eq:linear_perceptron} corresponds to the standard linear regression model when $\nu=1$ (ie. only one output unit). In equation \eqref{eq:linear_perceptron} the output unit simply linearly summed up the weighted inputs, but the sum may be transformed by a so called \emph{activation function}. This yields the \textit{single unit perceptron}, which with the right activation function is akin to a generalized linear model (GLM). What differentiates ANNs from GLMs is that they usually contain "hidden" layers. That is, arranged layers of response units between the input and output layers, that are not directly observable. These hidden layers treat the output of the previous layer as their inputs and induce nonlinearities by transforming the input through an activation function. For the theoretical considerations we treat the model with a single hidden layer and a single ouput unit (ie. $\nu=1$), 

\begin{equation}\label{eq:single_hidden_nn}
	f(x, \theta)= F\left(\beta_{0} + \sum_{k=1}^{m} G(\tilde{x}^\top\gamma_k)\beta_{k}\right),
\end{equation}

where $m$ is the number of the so called hidden units (processing units as the $\nu$ in \eqref{eq:linear_perceptron}) in the hidden layer. $\beta_k, \; k=0,...,m$ are the connection weights from hidden unit $k$ to the (in this case single) output unit, $\theta=(\beta_0,\dots, \beta_m, \gamma_1^\top,...,\gamma_m^\top)^\top$ is the vector of all network weights and $G$ is the activation function. For ease of notation we will set $F$ equal to the identity function (which is generally done when faced with a regression task) unless otherwise stated. This means it will be omitted from further notation. Figure \ref{fig:single_layer_NN} illustrates the architecture of a GLM and a single hidden layer MLP to underline their similarities. 

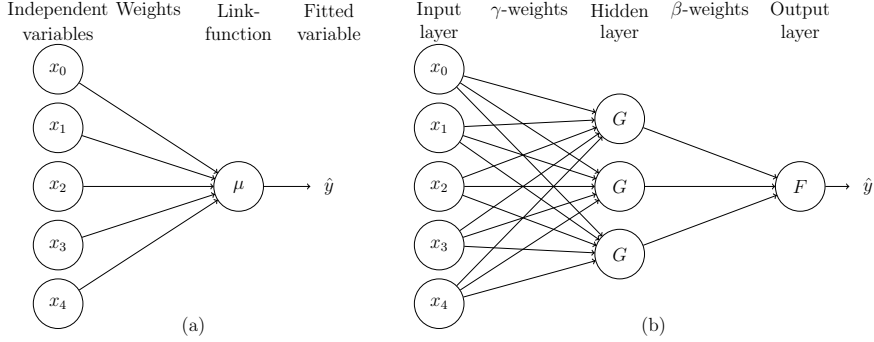
\begin{figure}
	\begin{center}
		\scalebox{0.6}{\begin{tikzpicture}

\tikzstyle{place}=[circle, draw=black, minimum size = 11mm]

\draw node at (0, -1.3) [place] (first_1) {$x_0$};
\draw node at (0, -2*1.3) [place] (first_2) {$x_1$};
\draw node at (0, -3*1.3) [place] (first_3) {$x_2$};	
\draw node at (0, -4*1.3) [place] (first_4) {$x_3$};
\draw node at (0, -5*1.3) [place] (first_5) {$x_4$};		

\draw node at (4, -3*1.3) [place] (second_2) {$\mu$};

\draw node at (6, -3*1.3) [circle] (fourth_2) {$\hat{y}$};

\foreach \i in {1,...,5}
\foreach \j in {2,...,2}
\draw [->] (first_\i) to (second_\j);

\foreach \i in {2,...,2}
\foreach \j in {2,...,2}
\draw [->] (second_\i) to (fourth_\j);

\node at (0, 0) [black, ] {Independent};
\node at (0, -0.5) [black, ] {variables};
\node at (2, 0) [black, ] {Weights};
\node at (4, 0) [black, ] {Link-};
\node at (4, -0.5) [black, ] {function};
\node at (6, 0) [black, ] {Fitted};
\node at (6, -0.5) [black, ] {variable};
\node at (3, -7)[black, ] {(a)};

\end{tikzpicture}
\qquad
\begin{tikzpicture}

\tikzstyle{place}=[circle, draw=black, minimum size = 11mm]

\draw node at (0, -1.3) [place] (first_1) {$x_0$};
\draw node at (0, -2*1.3) [place] (first_2) {$x_1$};
\draw node at (0, -3*1.3) [place] (first_3) {$x_2$};	
\draw node at (0, -4*1.3) [place] (first_4) {$x_3$};
\draw node at (0, -5*1.3) [place] (first_5) {$x_4$};

\draw node at (4, -2*1.2) [place] (second_1) {$G$};
\draw node at (4, -3*1.3) [place] (second_2) {$G$};
\draw node at (4, -4*1.35) [place] (second_3) {$G$};

\draw node at (8, -3*1.3) [place] (fourth_2) {$F$};

\draw node at (9.5, -3*1.3) [circle] (output_2) {$\hat{y}$};

\foreach \i in {1,...,5}
\foreach \j in {1,...,3}
\draw [->] (first_\i) to (second_\j);

\foreach \i in {1,...,3}
\foreach \j in {2,...,2}
\draw [->] (second_\i) to (fourth_\j);

\foreach \i in {2,...,2}
\draw [->] (fourth_\i) to (output_\i);

\node at (0, 0) [black, ] {Input};
\node at (0, -0.5) [black, ] {layer};
\node at (2, 0) [black, ] {$\gamma$-weights};
\node at (4, 0) [black, ] {Hidden};
\node at (4, -0.5) [black, ] {layer};
\node at (6, 0) [black, ] {$\beta$-weights};
\node at (8, 0) [black, ] {Output};
\node at (8, -0.5) [black, ] {layer};
\node at (4.75, -7) [black, ] {(b)};

\end{tikzpicture}}
		\caption[Schematic pepresentation of two perceptrons]{(a): A single unit perceptron (GLM). (b): A single hidden layer ANN}
		\label{fig:single_layer_NN}
	\end{center}
\end{figure} 

Models with more hidden layers can then be understood as a chained version of \eqref{eq:single_hidden_nn} where successive hidden layers consider the output of the previous hidden layer as inputs. This nested nature of more complex models made gradient descent the method of choice for the estimation. Unfortunately though, this practically limits the choice for the activation function. In previous theoretical research the discussion was mostly omitted, as the universal approximation property (see the subsequent section \ref{th:approximation_property}) can be employed by most of them. Whereas the logistic function was mostly used in earlier applications, newer deeper networks mostly rely on the Rectified Linear Unit (ReLU) activation function which is defined as:
$$
f(x) = \max(0,x)
$$

The reason for this is a property called the \emph{vanishing gradient problem}. For the moment, we only considered single layer networks, but just by considering Figure \ref{fig:single_layer_NN} and Equation \eqref{eq:single_hidden_nn} the problem can be seen easily. In the backward pass of the gradient descent, the gradient will also depend on the derivative of the activation function. If there are multiple intermediate layers, the gradient will naturally multiply. In the case of the (formerly popular) logistic function, the derivative peaks at $0.25$. In deeper networks, the effect is exacerbated and the gradient quickly becomes too small for any meaningful optimization. ReLU on the other hand has a constant derivative and results in more stable results. Theoretical results have already been derived for the logistic function and the radial basis function but to the best of our knowledge not for the ReLU activation function. In what follows, we will extend the general results from eg. \cite{chen_racine_2001} to include ANNs that make use of the ReLU activation function.  

\subsection{Universal Approximation Property}\label{th:approximation_property}

The reason ANNs are very flexible and adapt to many problems is that they possess the \textit{universal approximation property}. Among others, \cite{Hornik_approximation} have established, that ANNs with suitably many hidden units and layers can approximate any function arbitrarily well. More recently \cite{approximation_relu} established the universal approximation theorem also holds for single hidden layer ANNs with unbounded activation functions such as the ReLU. To make things clearer also throughout the next chapters, we adopt the definition of denseness and convergence from \cite{Hornik_approximation}:

\begin{definition}[Denseness and uniform convergence] \label{def:denseness}
	A subset $S$ of a metric space $(X, \rho)$ is $\rho$-\textit{dense} in a subset $T$ if for every $\varepsilon > 0$ and for every $t \in T$ $\exists s \in S$ such that $\rho(s,t) < \varepsilon$. To establish the convergence for functions, we define further: Let $C^d$  be the set of continuous functions from $\mathbb{R}^{d}$ to $\mathbb{R}$ and $K \subset \mathbb{R}^{d}$ a compact subset. For $f,g \; \in C^d$ let $\rho_K(f,g)\equiv\sup_{x \in K}|f(x) - g(x)|$. A subset $S$ in $C^d$ is then said to be \textit{uniformly dense on compacta} in $C^d$ if for every $K$, $S$ is $\rho_K$ dense in $C^d$. Further, a sequence of functions $\{f_n\}$ converges to a function $f$ uniformly on compacta if for all (compact) $K\subset \mathbb{R}^{d}$, $\rho_K(f_n,f) \rightarrow 0$ as $n \rightarrow \infty$
\end{definition}

The following theorem then summarises the findings of \cite{Hornik_approximation} and \cite{approximation_relu}:

\begin{theorem}[Universal approximation theorem] \label{theorem:universal_approx_theorem}
	Let $G(\cdot)$ be either a squashing function (mapping arbitrary input to a real (nondecreasing) bounded value) or an unbounded rectified linear unit (ReLU) activation function. An ANN of the form of \eqref{eq:single_hidden_nn} is, for every $d \in \mathbb{N}$ and $G(\cdot)$, uniformly dense on compacta (\cite{Hornik_approximation} theorem 2.2). 
\end{theorem}

Further, for squashing (\cite{Hornik_approximation}) and ReLU (\cite{approximation_relu}) activation functions, the ANN from \eqref{eq:single_hidden_nn} converges uniformly on compacta. For ANNs with ReLU activations the number of hidden units is finite. 

This establishes the asymptotic approximation properties but still leaves open the choice of the number of hidden units for a concrete problem. In theory, "enough" hidden units will ensure that the estimator can approximate any function. For example, we can simulate data according to:

\begin{equation}\label{apB:model_chaos}
	y_{t} = 0.3 y_{t-1} + \frac{22}{\pi}\sin(2\pi y_{t-1} + 0.\overline{33})
\end{equation}

and fit an ANN on the points. Figure \ref{fig:mseperunitsplot} illustrates the approximation property of the single-hidden layer ANN with an increasing number of hidden units. Note though that the model did not contain any random variation, which means there is no ''downside'' to adding more hidden units. In practice, with noisy data, there is a tradeoff between the bias and variance of an estimation, which can be derived with the help of sieve estimation. 

\begin{figure}\label{fig:mseperunitsplot}
	\centering
	\includegraphics[width=0.7\linewidth]{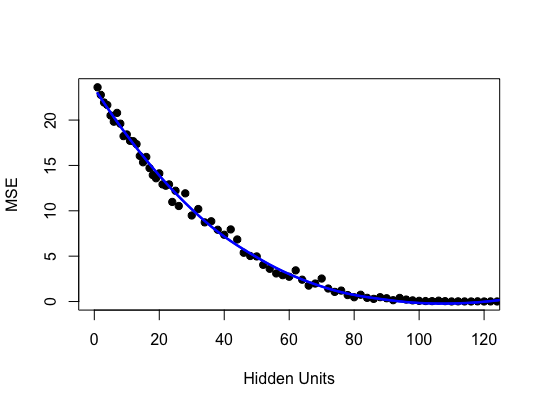}
	\caption[MSE with increasing number of hidden units]{Illustration of the decreasing Mean Squared Error (MSE) with increasing number of hidden units. All simulations were calculated with 200 epochs (passes of the dataset). Note that the MSE does not decrease strictly with the number of hidden units. This is an indicator that the gradient descent algorithm got stuck in a local minimum.}
\end{figure}

\subsection{Sieve Estimation}\label{th:sieve_estimation}

To establish convergence results, the ANN models can be incorporated into the theory of \emph{sieve estimation}. The method is generally attributed to \cite{grenander_1981}.  In semi- or nonparametric regression settings, the set of parameters  $\theta$ that describe a certain process is situated in a not necessarily finite parameter space, which makes the optimization problem no longer well posed (\cite{handbook_econometrics_chen}). The method of sieves approximates the infinite parameter space with a sequence of approximating parameter spaces that will be dense in the original (possibly infinite) space in the limit. To illustrate the approach we take a generic time series model of the form:

\begin{equation}\label{eq:generic_timeseries_model}
	y_t = \psi(x_{t}) + u_t, \quad t=1,...,n,
\end{equation}

where $x_{t}$ is a $q \times 1$ vector (with possibly lagged values of $y_t$ and $u_t$ is random noise with $\mathbb{E}[u_t|x_t]=0$ and $\mathbb{E}[u_t^2|x_t]=\sigma^2$. If we take the least squares approach, the optimal estimator can be written as:

$$
\hat{\psi}(\cdot) = \arg\min_{\psi(\cdot) \in \Psi}\frac{1}{n}\sum_{t=1}^{n}(y_t - \psi(x_{t}))^2
$$

In parametric settings the functional space of $\Psi$ is typically strictly limited to an additive form with the parameter vector bound in some space $\mathbb{R}^q$. Grendanders idea consists of estimating $\psi$ by a parametric class of increasingly complex sieve-functions, in the case for the neural network defined in \eqref{eq:single_hidden_nn} that is:

\begin{equation}\label{eq:sieve_functionspace}
	\psi \in \Psi = \bigg\{f: f(x; \theta)= \sum_{k=1}^{r_n}\beta_k G(\tilde{x}^\top \gamma_k) \bigg\}\;,
\end{equation}


where the activation functions $G$ are a sequence of basis functions and $r_n$ denotes the number terms in the sieve space when we have a sample of size $n$. Given \ref{th:approximation_property}, $f$ can then approximate the true function $\psi(\cdot)$ arbitrarily well with increasing $r_n$\footnote{This of course is not limited to ANN type sieves, for example, with $G$ a polynomial function of degree $r_n$ the Stone-Weierstrass theorem can be employed directly.}. The idea is then, to let $r_n$ increase with $n$, this feature distinguishes the estimation technique from standard parametric techniques, which assume a fixed and finite dimensional parameter space. The function $f$ is still parametrized by $\theta \in \Theta_n$ (the sieve-space). To properly describe $\Theta_n$ we define the infinite dimensional parameter space of the nonparametric problem as $\Theta$ with a metric $d$ and its corresponding approximation space $\Theta_n$. To have well defined properties, we let the approximation spaces be compact and nondecreasing in $n$ and all should be $\subseteq \Theta$. The sieve estimator is then given by:

\begin{center}
	\begin{equation} \label{eq:sieve_estimator}
		\hat{\psi}_n(\cdot) = \arg\min_{\theta \in \Theta_n}\frac{1}{n}\sum_{t=1}^{n}(y_t - \psi_n(x_{t}))^2 \;,
	\end{equation}
\end{center}

with the definition of $\psi_n$ from above. This then becomes a feasible estimation problem because in the parameter space of $f(\cdot, \theta)$ the optimization is well defined. Note that in order to achieve consistency, we need the following assumption to hold (adapted from \cite{newey_1997} for series estimation):

$$
\frac{1}{r_n} \rightarrow 0 \; \text{and} \; \frac{r_n}{n} \rightarrow 0\text{ as } n \rightarrow \infty
$$

Which can be closely compared to the conditions of nonparametric kernel regression. Also note that the requirement on the left hand side ensures that the parameter space is infinite in the limit of $n$ but grows slower than the sample size (right condition). The simple but powerful form also shows why ANNs are a popular choice in modelling financial data. Restrictions such as additivity or previous knowledge on the distribution of errors, such as excess kurtosis, can be imposed (cf. \cite{handbook_econometrics_chen}) just as well as certain shape constraints such as concavity (cf. \cite{handbook_financial_torben} and could be used to improve the resulting model. Due to its parametric form, it also yields parameter estimates, which will be useful if the results are to be (at least partially) interpretable. Formally, the estimation process with ANNs can then be described (as for example in \cite{handbook_econometrics_chen}):

\begin{definition}[Approximate Sieve Extremum Estimate]
	Define $Q$ as some population criterion function $Q:\Theta \rightarrow \mathbb{R}$, which is uniquely maximised at $\theta_0 \in \Theta$. $\widehat{Q}_n:\Theta \rightarrow \mathbb{R}$ is its sample equivalent based on the data $\{X_t\}_{t=1}^n$, such that it converges to the true $Q$ with $n \rightarrow \infty$. Then, $\hat{\theta}_n$ is defined as the approximate maximizer of $\hat{Q}_n(\theta)$ over the sieve space $\Theta_n$. that is:
	
	\begin{equation}\label{eq:approximate_sieve_est}
		\widehat{Q}_n(\hat{\theta}_n) \geq \sup_{\theta \in \Theta_n}\widehat{Q}_n(\theta) - \mathcal{O}(\eta_n), \; \textnormal{with} \; \eta_n \rightarrow 0 \; \textnormal{as} \; n \rightarrow \infty \;,
	\end{equation}
	where $\eta_n$ denotes an approximation error in the sample criterion. $\mathcal{O}$ denotes the standard big-O notation\footnote{$f(n)=\mathcal{O}(a_n)$ simply means $\exists M < \infty$ s.t. $|f(n)| \leq Ma_n, \forall n \in \mathbb{N}$ (eg. the data set size).}.
\end{definition}

 Put simply, this means that the estimated maximiser over the sieve space converges to the supremum of the sample criterion over the sieve space. The equation can be interpreted as a simplification of the original problem that was situated in the (possibly infinite) dimensional parameter space $\Theta$ and now belongs to the (finite but variable - sized) parameter space $\Theta_n$. With the above stated condition that the sieve spaces are nondecreasing, we use the projection $\pi_n$ and define $\pi_n\theta_0 \in \Theta_n$ such that $\varphi_n \equiv d(\theta_0, \pi_n\theta_0) \rightarrow 0$ as $n \rightarrow \infty$, which is the \textit{approximation error} due to the lower dimension of the sieve space compared to the original parameter space. Defining the sample equivalent with $n$ observations of $y_t\in\mathbb{R}$, $x_t \in \mathbb{R}^d$ and the function $l=-(y_t - \psi(\theta, x_t))^2$, we can rewrite the optimization problem as:

\begin{center}
	\begin{equation}\label{eq:sieve_m_formula}
		\sup_{\theta \in \Theta_n}\widehat{Q}_n(\theta)=\sup_{\theta \in \Theta_n}\frac{1}{n}\sum_{t=1}^{n}l(\theta, (y_t, x_{t}')')
	\end{equation}
\end{center}

This specifies the sieve estimator as an M-estimator, which allows us to give some interpretation to the loss function. Instead of just maximizing the objective function (which would mean minimizing the loss), we are maximising the empirical criterion function $\widehat{Q}_n$ over the (simpler) approximation space $\Theta_n \subset \Theta$.

\subsection{Sieve Order} \label{th:sieve_order}

When we consider the functional sieve form of equation \eqref{eq:sieve_functionspace} it becomes clear that the ANN of equation \eqref{eq:single_hidden_nn} can be stated as a sieve estimator (by keeping $F(\cdot)$ the identity and letting the number of hidden units grow with $n$\footnote{To make the change to a variable number of hidden units explicit, we will adopt the term $r_n$ to indicate the number of hidden units in line with our considerations about sieve estimation from above.}). Given the universal approximation property from theorem \ref{theorem:universal_approx_theorem}, we also know that in the limit, ANNs can approximate any function, which guarantees that the sieve problem is well defined. The choice of the sieve order $r_n$ then becomes the choice of the number of hidden units in the ANN\footnote{As will be discussed later on, the number of hidden units can be regularized. For the theoretical considerations the number of hidden units can be considered the number of non-zero hidden units}. As it is the case across all nonparametric methods, ANN sieves also have a bias-variance tradeoff, that can be controlled by the number of hidden units (or more formally, by the number of sieve terms). Although sieve estimators are parametrized, they approximate a function of unknown form. We will therefore first illustrate the tradeoff using the mean integrated squared error (MISE) to make the relation to other nonparametric function approximations clear. For a function with variable $x$ with density $f_X$ the definition of the MISE is:

$$
\text{MISE}=\mathbb{E}||f_n - f||_2^2=\mathbb{E}\int \big(f_n(x)- f(x)\big)^2 dx\;,
$$

where $f_n(x)$ is the ANN approximation of the sieve order $r_n$ and $f$ the function of interest. To derive the following proposition we use the property of ANNs that they represent the solution of a linear equation in the output layer. This allows us to exploit some known properties of linear models. For example, we can assume that the solution to the optimization problem in the final hidden layer can be found with a standard linear technique. Therefore, we define $g_{k}(x)$ the (nonlinearly) transformed ($n \times 1$) output produced by the $k^{th}$ unit in the final hidden layer and $\mathcal{G}(x)$ the ($n \times r_n $) corresponding matrix of all $r_n$ units and $\mathcal{G}(x_t)$ the $1 \times r_n$ transformed feature vector for observation $t$. We can then follow the methodology of \cite{hansen_2014_hb} who considers spline based sieves and define our model from equation \eqref{eq:generic_timeseries_model} and the corresponding (ANN) estimation as:

$$
\psi_n =  f_n(x,\theta)=\beta_0 + \sum_{k=1}^{r_n}\beta_kg(x)=\mathcal{G}(x)\beta_n \;,
$$
$$
\hat{\beta}_{n}= (\mathcal{G}(x)'\mathcal{G}(x))^{-1}\mathcal{G}(x)'Y \;,
$$

where we include the $\gamma$ weights implicitly as defined above and allow the first column of $\mathcal{G}(x)$ to be a constant. Also note that we indexed $\beta_n$ because the parameter-vector will change its size with increasing sieve order. Further, contrary to the standard linear case, where the only error is the projection error, we have an additional \emph{approximation error} in $\hat{\beta_n}$. We define the approximation error and its squared expectation as:

\begin{equation}\label{eq:definition_approximation_error}
	\upsilon_n(x_t)=\psi(x_t) - \mathcal{G}(x_t)\beta_n \;,
\end{equation}
$$
\varphi_{n}^2 = \mathbb{E}[\upsilon_n(x_t)^2]=\int \upsilon_n(x)^2f_X(x)dx
$$

Taking this together and following the methodology of \cite{hansen_2014_hb}, we are then able to propose the following (proof in appendix \ref{apA:bias_variance}):

\begin{proposition}[Bias-Variance Tradeoff for ANN Sieves]
	
	For model \eqref{eq:single_hidden_nn} with the identity activation function for the output layer we have a bias-variance tradeoff which can be summarized by
	$$
	\text{MISE}_n(x)=\mathbb{E} \int \big (\hat{\psi_n}(x) - \psi(x) \big )^2dx = \varphi_{n}^2 + n^{-1}trace(\mathcal{E}_{n}^{-1}\Omega_n) 
	$$ 
	$$
	\mathcal{E}_n = \mathbb{E}[\mathcal{G}(x_t)'\mathcal{G}(x_t)], \;\; \Omega_n = \mathbb{E}[\mathcal{G}(x_t)'\mathcal{G}(x_t)\sigma_t^2], \;\; \sigma_t^2 = \mathbb{E}[u_t^2|X_t]
	$$
	\label{theorem:bias_variance}
\end{proposition}

\begin{figure}
	\centering
	\includegraphics[width=0.6\linewidth]{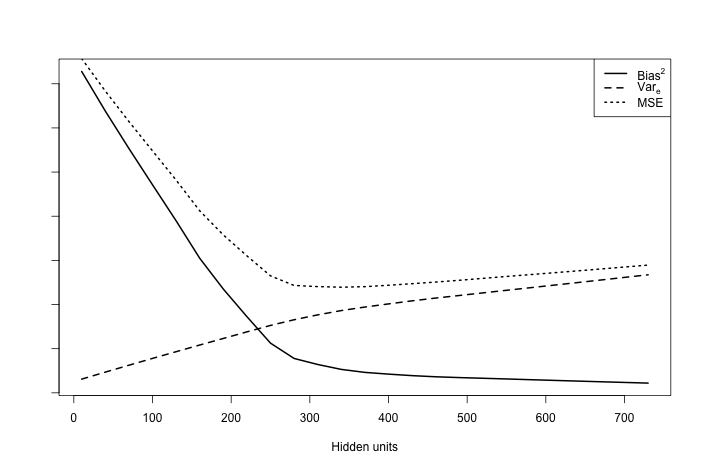}
	\caption[Bias-variance tradeoff for ANN sieves]{Pointwise MSE decomposition of an ANN Sieve estimator, with increasing number of hidden units, we can observe the standard bias-variance tradeoff for nonparametric methods. Data was generated from model \ref{apB:iid_model}.}
	\label{fig:bias_variance_decomposition}
\end{figure}

Given that we used a linear model, we note that for the single layer ANN, models with more hidden units will nest smaller (in terms of hidden units) models. If we then consider $r_n'>r_n$ the optimization problem for sieve order $r_n$ implicitly poses $\beta_k = 0$ for all $k=r_n+1,...,r_n'$ in the final hidden layer. Therefore, the expected approximation error will be at least weakly decreasing in the number of sieve terms $r_n$ (assuming the maximum of $\widehat{Q}_n$ is found) and hence $\varphi_{n}^2$ will decrease with increasing $n$. This is a well known property when working with the standard linear model. The second term corresponds to the asymptotic variance and the term in the trace will (weakly) increase in $r_n$, as the diagonal is weakly positive. Figure \ref{fig:bias_variance_decomposition} depicts the decomposed mean squared error\footnote{Which can be expressed as $\text{MSE}=\mathbb{E}[(f(x)-\hat{f}(x))^2]=\text{Bias}\hat{f}(x)^2 + \text{Var}(\hat{f}(x))$. For a derivation see for example \cite{geman_1992}.}. Given that ANNs were established as sieve estimators, an explicit bound can then be derived.

With the bias-variance tradeoff dependent on the number of sieve terms, $r_n$ becomes a hyperparameter that is similar to the choice of for example the bandwidth in kernel-type estimators. Similar to the optimal choice of the bandwidth in such estimators, the optimal number of $r_n$ is dependent on the true underlying function and hence we are not able to get an exact result with noisy data from an unknown distribution. Unfortunately and contrary to the case with kernel methods, there is no rule of thumb on how many sieve terms should be chosen. The from of the ANN sieve allows for regularization though. For example, $L_1$-regularization will effectively set some weights and some hidden nodes to zero. Instead of providing the network with the (unknown) correct number of sieve terms, we can simply provide a sufficiently large number of hidden nodes and let the algorithm choose the optimal non-zero number via eg. cross validation or out of batch prediction. This is often computationally less demanding than other procedures employed for finding the optimal terms for series estimation (such as the \textit{jackknive} of eg. \cite{hansen_2014_hb})

\subsection{A Partially Linear Model}\label{th:partially_linear_model}

In light of the finding in \cite{makridakis2020m4} that statistical approaches and Machine Learning methods seem to be complements rather than substitutes, we now propose a semiparametric model. A popular way to state such a semiparametric model is as a \textit{partially linear model}, which contains both a (finite dimensional) parametric part and a (possibly infinite) nonparametric part. In line with most financial data, the convergence results later on are derived for time-series data. As it is common throughout the literature, we impose that the data are only weakly dependent and follow the so-called $\beta$-mixing process (the i.i.d. case would be similar). 

\begin{definition}[$\beta$-mixing process]\label{def:beta_mixing}
	Let $\mathcal{M}_t^{t+\tau}$ denote the $\sigma$-field generated from the stochastic series ${Z_s}_{s=t}^{t+\tau}$. Define:
	$$
	b_{\tau} = \sup_{t \in \mathbb{N}} \sup_{A \in \mathcal{M}_{t+\tau}^{\infty}, B \in \mathcal{M}_{-\infty}^{t}} |\mathbb{E}[A|B] - \mathbb{E}[A]|
	$$
	The sequence ${Z_t}_{t=-\infty}^{\infty}$ is said to be $\beta$-mixing if $b_{\tau} \rightarrow 0$ as $\tau \rightarrow \infty$
\end{definition}

The mixing condition is imposed to ensure that the statistical dependence between two data points in time goes to zero as the temporal distance between them increases. The $\beta$-mixing condition is used throughout the literature because it permits to obtain faster convergence rates than for strong ($\alpha$)-mixing processes but still allows for a rich variety of economic time series (cf. \cite{chen_racine_2001}). Given the $\beta$-mixing condition, a partially linear model  can then be interpreted as a parametric model with time-varying intercepts driven by the nonparametric component (\cite{applied_ghysels_2018}). This specific form allows us to give a meaning to the nonparametric part of the equation as well. For a dependent variable $Y_t$ and a vector of covariates $M_t$, the formulation is:

\begin{equation}\label{eq:partially_linear}
	Y_t = \mathbb{E}[Y_t|M_t] + e_t = m(M_t) + e_t = x_t'\beta + \phi(z_t)  + e_t \quad t=1,...,n
\end{equation}

where the vector $M_t$ is partitioned into covariates used in the linear and nonlinear model components ($x_t$ and $z_t$ respectively). Restrictions can easily be imposed on the model. A popular choice is to impose additivity on the $z_t$ terms in the nonparametric component, in order to reduce the dimensionality of the problem (ie. set $\phi(z_t)=\phi(z_{1,t}) + \phi(z_{2,t}),...,\phi(z_{k,t})$). The nonparametric part of the equation then becomes the generalized additive model (\cite{hastie1990generalized}). The issue with such a restriction is that they require the model to be correctly specified in the case of (nonlinear) interactions between the input variables. Instead, we can also specify the partially linear model as an augmented ANN (cf. \cite{kuan_white_1994})

\begin{equation}\label{eq:augmented_nn}
	f(x, \theta)= x' \alpha + \beta_{0} + \sum_{j=1}^{r_n} \beta_{j} G(\tilde{x}'\gamma_j)
\end{equation}

Note that the $\theta$ vector now contains the $\alpha$ as well. Further, note that we did not impose a special form on the $\alpha$ vector. This allows us to include parametric nonlinear models as special cases. The model also shows a further advantage of using sieve estimators. Because we treat the nonparametric component as a parametric extension of an existing model, we can include it easily, even if the existing model is nonlinear (eg. in a logistic GLM - which could be used in a latent variable model). As discussed in the introduction, heuristically there seem to be better approximations properties in multilayer networks, which is why we propose to model the partially linear model from \eqref{eq:partially_linear} with a multilayer ANN-sieve as the nonparametric component:

\begin{equation}\label{eq:two_layer_ann}
	\phi_n(z_t; \theta)= \vartheta_0 + \sum_{r=1}^{k'}\vartheta_r G \big ( \sum_{j=0}^{k}\beta_{j}G(\tilde{z_t}'\gamma_j)+ \iota_r \big)
\end{equation}

Where we set $G(\tilde{x}'\gamma_0)=1$. The vector $\theta$ now contains $(d+1)k + (k'(k+1)) + (k'+1)$ parameters. The parameter $d$ represents the dimension of $z_t$ (the dimensionality used in the nonparametric part). The covariate vector $x_t$ can include lagged values of $Y_t$, such that the parametric part might be specified as a standard parametric autoregressive (AR) model. The combined model then becomes:

\begin{equation}\label{eq:semi_multi_ann}
	Y_t = x_t'\zeta + \vartheta_0 + \sum_{r=1}^{k'}\vartheta_r G \big ( \sum_{j=0}^{k}\beta_{j}G(\tilde{z_t}'\gamma_j)+ \iota_r \big) + e_t
\end{equation}

We note that we did not include a constant in the parametric part. This is due to the general fact that in a model of the form of \eqref{eq:partially_linear} the parametric should not contain a constant. Because the functional form of $\phi(z_t)$ is not specified, it can incorporate the constant which would make an identification impossible (\cite{nonparametric_econ_li}).\footnote{A closer inspection of \eqref{eq:semi_multi_ann} reveals that the constant will be absorbed by the bias unit in the final layer of the ANN.}

This model is extendable to even more than two hidden layers if it is desired. As described above, multilayer ANNs with eg. sigmoid type activation functions suffer from the vanishing gradient problem. Therefore, we propose the use of the ReLU function as nonlinear transformation. We will henceforth call the proposed model the SANN (semiparametric Artificial Neural Network). Figure \ref{fig:semi_hybrid_model_drawing} depicts the SANN in the network architecture as already seen before. The subsequent subsection discusses the convergence properties of the proposed estimator.

\begin{figure}
	\begin{center}
		\scalebox{0.6}{\begin{tikzpicture}

\tikzstyle{place}=[circle, draw=black, minimum size = 11mm]

\draw node at (0, -1.3) [place] (first_1) {$x_1$};
\draw node at (0, -2*1.3) [place] (first_2) {$x_2$};
\draw node at (0, -3*1.3) [place] (first_3) {$x_3$};
\draw node at (0, -4*1.3) [place] (first_4) {$x_4$};



\draw node at (3, -3.2*1.3) [place] (second_1) {$G_1$};
\draw node at (3, -5*1.3) [place] (second_2) {$G_k$};

\path (second_1) -- (second_2) node [black, font=\Huge, midway, sloped] {$\dots$};

\draw node at (7, -3.2*1.3) [place] (hidden_1) {$G_1$};
\draw node at (7, -5*1.3) [place] (hidden_2) {$G_{k'}$};

\path (hidden_1) -- (hidden_2) node [black, font=\Huge, midway, sloped] {$\dots$};

\draw node at (10, -1.3) [place] (linear_0) {$F$};

\draw node at (12, -1.3) [circle] (output_2) {$\hat{y}$};


\foreach \i in {3,...,4}
\foreach \j in {1,...,2}
\draw [->] (first_\i) to (second_\j);

\foreach \i in {1,...,2}
\draw [->] (first_\i) to (linear_0);

\foreach \i in {1,...,2}
\foreach \j in {1,...,2}
\draw [->] (second_\i) to (hidden_\j);

\foreach \i in {1,...,2}
\draw [->] (hidden_\i) to (linear_0);

\draw [->] (linear_0) to (output_2);

\node at (0, 0) [black, ] {Input};
\node at (0, -0.5) [black, ] {layer};
\node at (5, 0) [black, ] {Hidden};
\node at (5, -0.5) [black, ] {layers};
\node at (10, 0) [black, ] {Output};
\node at (10, -0.5) [black, ] {layer};

\end{tikzpicture}}
		\caption[Schematic drawing of SANN model]{Representation of a semi-(non)parametric augmented ANN with two hidden layers and two input variables each. Note that we are not restricted to separate the linear and nonlinear inputs.}
		\label{fig:semi_hybrid_model_drawing}
	\end{center}
\end{figure}
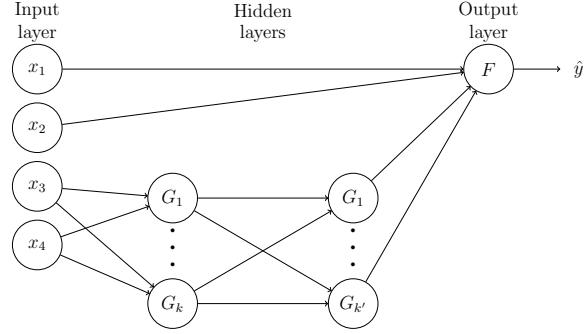

\subsection{Asymptotic Properties}\label{th:asymptotic_properties}

By employing the more general framework of sieve estimation, the asymptotic behaviour of the SANN can be derived. The interest here is twofold: First, convergence rates permit to compare the SANN sieve to other nonparametric methods which could help explain their popularity. And second, the rates will provides guidance of how fast the estimator will converge to the true function even if we only consider the nonparametric part as a nuisance term. 

\subsubsection{Convergence of the Parametric Part}

First we will discuss the convergence of the parametric part of the regression equation in \eqref{eq:semi_multi_ann}. We propose the use of a two-step procedure to estimate the model because it has several advantages. To begin with, it makes it easier to derive properties for the inference on the parametric part because we can use existing theory, and in addition, it makes use of the specific form of our ANN sieve. Note that we will follow the notation \eqref{eq:partially_linear} to keep the expressions concise (hence the derived $\beta$ corresponds to $\zeta$ in equation \eqref{eq:semi_multi_ann}). 

Although we have a linear and additive structure in the output layer, but the preceding hidden layer ensures that the nonparametric part is nevertheless approximated well by the ANN terms.\footnote{The next subsection will treat the convergence of the nonparametric part. For the moment, just suppose the estimates in the output layer are sufficiently close to the true values.} The idea is then to isolate the parametric component and estimate it with standard procedures. The simplest form to isolate $\beta$ would be by taking conditional expectations and differencing out from model \ref{eq:partially_linear}:

$$
Y_t - \mathbb{E}[Y_t|z_t] = (x_t -\mathbb{E}[x_t|z_t])'\beta + (e_t - \mathbb{E}[e_t|z_t])
$$
Then we define:
$$
\tilde{Y_t}=Y_t - \mathbb{E}[Y_t|z_t], \quad \tilde{x_t}=x_t - \mathbb{E}[x_t|z_t]
$$
Which allows us to provide the least squares solution as.
\begin{equation}\label{eq:infeasible_beta}
	\hat{\beta}_{\text{infeasible}} = (\tilde{X}'\tilde{X})^{-1}\tilde{X}'\tilde{Y}\;,
\end{equation}

where we used the standard matrix notation. We call this the infeasible estimator, as it would require the unknown quantities $\mathbb{E}[Y_t|z_t]$ and $\mathbb{E}[x_t|z_t]$. Because we can estimate the nonparametric part with our ANN, we are nonetheless able to derive a consistent estimate $\hat{\mathbb{E}}$. The specific form of the final layer makes the isolation of the parametric part convenient. By using partitioned regression (see appendix \ref{apA:parametric_part} for details) we are able to obtain the estimator for the parametric part:

\begin{equation}\label{eq:parametric_estimator}
	\hat{\beta} = (\tilde{X}^{*'}  \tilde{X}^{*} )\tilde{X}^{*'}  \tilde{Y}^{*}
\end{equation}
$$
\tilde{X}^{*}=X-\hat{\mathbb{E}}[X|Z], \quad \tilde{Y}^{*}=Y-\hat{\mathbb{E}}[Y|Z]
$$

If we impose a further condition, that the convergence rate of the nonparametric part is \textit{slower} than $n^{\frac{1}{2}}$, which is shown in the next section. Then under some fairly general restrictions, it can be shown, that the parametric part of a sieve M estimate is asymptotically normal and converging at the rate $\sqrt{n}$. As we have already established, that our ANN sieve is an M-estimator and that we estimate $\beta$ with a two step procedure, we refer the interested reader to \cite{handbook_econometrics_chen} theorem 4.1 and only state the result:

\begin{theorem}[Asymptotic normality for semiparametric two-step estimation]
	$$
	\sqrt{n}(\hat{\beta}-\beta_0) \xrightarrow[\text{}]{\text{dist}} \mathcal{N}(0, \Omega)
	$$
	$$
	\Omega = (\Gamma_{1}'W\Gamma_{1})^{-1}\Gamma_{1}'WV_1W\Gamma_{1}(\Gamma_{1}'W\Gamma_{1})^{-1}\;,
	$$
	where $\Gamma_1$ is the partial derivative of of a function $M(\beta, \phi_0) = 0$ iff $\beta=\beta_0$ and $W$ a weighting matrix such that a sieve estimation $M_n$ weighted with W (ie $M_n(\beta, \phi_0)'WM_n(\beta, \phi_0)$) is close to $M(\beta, \phi_0)'WM(\beta, \phi_0)$. See \cite{handbook_econometrics_chen} for details. We also note that this is similar to the asymptotic distribution of standard series estimators (cf. \cite{nonparametric_econ_li}, Ch 15). The fact that the parametric part of our model converges to a normal distribution allows us to employ statistical inference on the estimated coefficients.
\end{theorem}

\subsubsection{Convergence of the Nonparametric Part}

What is left is the discussion of the convergence of the nonparametric of the SANN. For simplicity we will derive the upper bound for the single hidden layer ReLU ANN rather than the proposed multilayer network. This has two reasons. First, it allows us to directly compare the rates to the other ANN sieves that were discussed by eg. \cite{chen_racine_2001}. And second, a closer examination of the derivation in appendix \ref{apA:proof_convergence_sieves} reveals that the variance term in the trade-off function increases with additional parameters even if they are in differentl layers, without taking into account any potential benefit. This is due to the nature of the proof, where we derive a general \textit{upper bound} on the convergence. Nevertheless, the first point made above already merits a thorough discussion of the convergence rate of ReLU sieves. To get the convergence results, we first need to impose some technical assumptions (analogous to the assumptions imposed by \cite{chen_racine_2001}):

\begin{assumption}[Nature of the time series process]\label{assumption:beta_mixing}
	The considered time-series is stationary and $\beta$-mixing with $b_\tau\leq b_{0}\tau^{-\xi}, \xi > 2$
\end{assumption}

\begin{assumption}[Relationship between the parametric and nonparametric part] \label{assumption:non_mix_between_npandp}
	$\Delta x$ does not enter $\phi(z)$ additively and has full column rank. $z_t$ has compact support $\chi$ on $\mathbb{R}^{d}$
\end{assumption}

\begin{assumption}[Bounds on weights] \label{assumption:bounded_beta}
	We impose a limit on the weights of the $\beta$ coefficients in the ANN in equation \eqref{eq:semi_multi_ann}:
	$$
	\sum_{j=0}^{r_n}|\beta_j| \leq c_n, \quad \sum_{i=0}^{p} |\gamma_{i,j}| \leq c_{n_2},
	$$
	$$
	\sum_{i=0}^{p}|\gamma_{ij}x_i| \leq c_l, \implies \sum_{j=0}^{r_n}\beta_j\text{ReLU}(\gamma_{ij}x_i) \leq r_nc_nc_l 
	$$
\end{assumption}

Which ensures that our the function only maps to values within $\mathbb{R}$ and not $\bar{\mathbb{R}}$. Following \cite{HSWA_approximation} and \cite{handbook_econometrics_chen} we define a possible function space for $\phi$ in equation \eqref{eq:partially_linear}: Suppose that the true $\phi_0 \in \mathcal{P}\equiv \big \{ \phi \in L_2(\chi): \int_{\mathbb{R}^d}|\omega||\tilde{\phi}(\omega)|d\omega < \infty  \big \}$, where $\tilde{\phi}$ is the Fourier transform of $\phi$. That is, $\phi \in \mathcal{P}$ iff. it is square integrable and $\tilde{\phi}$ has finite first moments.

\begin{proposition}[Rate of convergence for ReLU-ANN sieves] \label{proposition:convergence_relu_sieves}
	Let $\hat{\theta}_n$ be the sieve M-estimate described by equation \eqref{eq:sieve_estimator}, $z_t \; \in \mathbb{R}^{d}$  and $\phi \in \mathcal{P}$. Suppose that assumptions \ref{assumption:beta_mixing}-\ref{assumption:bounded_beta} hold. By letting the sieve order grow according to $r_n^{2(1+\frac{1}{d+1})}\log(r_n)=\mathcal{O}(n)$ we achieve the following convergence rate for an ANN with the by assumption \ref{assumption:bounded_beta} modified ReLU activation function:
	$$
	||\hat{\theta}_n - \theta_0|| = \mathcal{O}\bigg([n/\log(n)]^{\frac{-(1+ \frac{2}{d+1})}{4(1+\frac{1}{1+d})}}\bigg)
	$$
\end{proposition}

Note that this corresponds to what \cite{chen_racine_2001} derived for ANNs using activation functions that suffer from the vanishing gradient problem. When following through the proof in appendix \ref{apA:proof_convergence_sieves}, we note the following. Given assumption \ref{assumption:bounded_beta}, we can express the sieve order for multilayer networks with $k$ units in the first layer and $k'$ units in the second layer as $r_n=k*k'$ and the implication in the assumption would stay the same. 

\subsection{Comparison to Kernel-Density Estimation}

Given that ReLU ANNs can achieve the same upper bound on the convergence, this motivates the use of ReLU activations for the proposed SANN in \eqref{eq:semi_multi_ann}. The established convergence results also permit the comparison of the asymptotic convergence rate of the SANN-sieve to that of other popular nonparametric estimators. For example, taking the optimal convergence for a kernel estimator (see eg. \cite{lecture_notes_nonparametric}):

$$
\text{AMISE}_{\text{opt}}=\mathcal{O}(n^{\frac{-2\textnormal{order}}{(2\textnormal{order}+\textnormal{dimension})}})
$$

\begin{figure}
	\centering
	\includegraphics[width=0.45\linewidth]{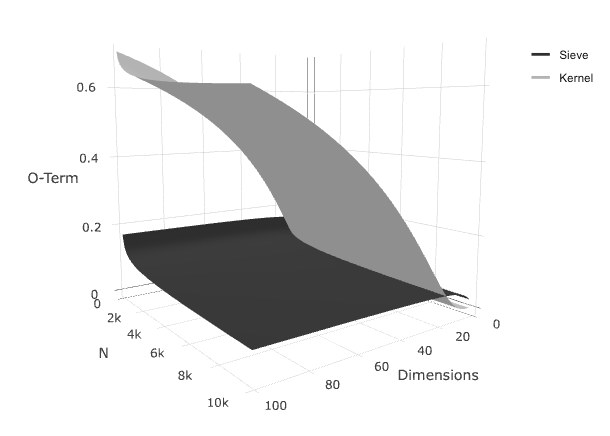}
	\caption[Optimal MISE kernel vs. sieve-type estimators]{Comparison of the optimal rate of convergence between multiple second order kernel- and sieve-type Estimators (for a single regressand).}
	\label{fig:convergence_sieve_kernel}
\end{figure}

When ignoring the constant, it becomes clear that in higher dimensional settings, the sieve estimator discussed above might have better convergence properties. Figure \ref{fig:convergence_sieve_kernel} provides a graphical illustration of the optimal rate of convergence for the two type of estimators (Where we used a second order kernel). The constant $M$ of our earlier definition of $\mathcal{O}$ is \textit{not} taken into account in this visualisation (although in the limit as $n \rightarrow \infty$ they will both converge to $0$ no matter the value of the constant). Therefore, we cannot compare the rates directly for finite samples. The illustration serves only to show that sieve estimators are less sensitive on the input dimension than kernel type estimators which in turn could explain some of the findings discussed in the introduction. To analyse how the convergence is affected in finite samples, Section \ref{sec:finite} provides results from finite-sample Monte Carlo studies. 

\section{Finite Sample Properties}\label{sec:finite}

The previous sections established ANNs as sieve estimators and as such, proposed an upper bound for the convergence. As the findings involved asymptotics, there might be large differences in finite samples. The present section compares ANNs to other nonparametric estimators by the means of mean squared error (MSE) and the help of simulations, given that the size of the data is finite. Definitions and formulas can be found in the Appendix. The simulations are run in the \texttt{R} Programming language and assess how the size of the dataset and the dimension (\texttt{N} and \texttt{d} in \eqref{proposition:convergence_relu_sieves} respectively) of the estimation affect the convergence as assessed by the MSE criterion. 

\subsection{Setup}

The simulations involving an ANN were done using the \texttt{keras} package (with a tensorflow backend). For ANN-type estimations, we chose enough epochs (passes of the dataset to the gradient descent algorithm) such that the loss levelled off. The optimization algorithm, number of hidden units and the regularization parameter were always chosen to be the same when comparing models. To compare the convergence results with regression, we use the package \texttt{np}. For the nonparametric kernel regression, we chose the local linear (LL) model with a second order gaussian kernel. Its theoretical convergence was briefly discussed in the previous section, its \textit{local} estimate has the form:

$$
\beta^{LL}(x_0)=(X'W_0X)^{-1}X'W_0Y\;,
$$

where $W_{0,i,i}=K(x_i,x_0,h)$ is the diagonal weighting matrix from the density estimates. When the \texttt{np} package is used, a bandwidth parameter is needed. We chose Silvermans adapted rule of thumb. For the univariate data it can be expressed as (eg. \cite{haerdle_sperlich_nonparametrics}):

$$
\text{bw}_{\text{silverman}}=\big(\frac{4 \hat{\sigma}^{5}}{3n}\big)^{\frac{1}{5}} \approx 1.06\hat{\sigma}n^{\frac{1}{5}} \;,
$$

where $\hat{\sigma}$ is the sample standard deviation. The \texttt{np} package calculates a slightly modified version for multivariate data: $\text{bw}_{\text{rule of thumb}} \approx 1.06 \sigma_j n^{\frac{1}{2P + l}}$. Where the $\sigma_j$ is calculated for every variable as $\min \{ \text{std.}, \text{mae.}/1.4826, \text{iqr.}/1.346 \}$, $P$ is the order of the kernel and $l$ the number of variables. We opted for the rule of thumb because it usually provides a good approximation to the optimal bandwidth and because other methods such as cross-validation (CV) often have impractically long calculation times for large datasets. Finally for the standard partially linear model a kernel-type estimator for the nonparametric part the package \texttt{PLRModels} was used. The package can only use a single variable as nonparametric component by default and does not provide a framework to implement out-of-sample predictions. 
For the simulations themselves, we specify the following setup: For each experiment, we will use $B=50$ MC iterations. Consider the general case of equation \eqref{eq:generic_timeseries_model}:

$$
Y_t = \psi(x_t)+ u_t, \quad t=1,...,n
$$

The considered $\psi(\cdot)$ can be either linear, partially linear or completely nonlinear. We then analyse the specific features of our ANN-sieve and the proposed model in equation \eqref{eq:semi_multi_ann}. For each simulation, we fix the data generating process (DGP) $\psi(\cdot)$ and leave it unchanged throughout the $B$ iterations. The noise term $u_t$ is drawn from a distribution and \textit{changed} for every iteration $b \; \in \{1,...,B\}$. Next, the relevant statistical metrics are calculated, which will change according to the specific issue analysed. Refer to appendix for details on those metrics. In general, we use pointwise metrics (calculated on the specific data-points) if we want to compare the performance of different estimators and integrated metrics when we compare the approximation properties. The integrated metrics are calculated by numerically integrating the pointwise metrics as suggested in \cite{handbook_econometrics_chen}.

\subsection{High Dimensional Estimation}\label{sim:high_dim}

The main property that we consider is the performance of ANN-sieves in a higher dimensional setting. We saw in Section \ref{sec:formal} that the $\mathcal{O}(\cdot)$ term for the ANN sieve is \textit{less} responsive to an increasing number of input parameters than that of kernel type estimators. From figure \ref{fig:convergence_sieve_kernel} we can visually infer that in settings where the input dimension is small, the kernel-type estimator seems to have a faster convergence. This advantage diminishes as the input dimension increases, which would be in line with the curse of dimensionality. We investigate the performance of the two estimators by generating data according to:

$$
y_i = \sum_{i=1}^{k}f_k(X_i) + \varepsilon_i, \quad i=1,...,500,\; k=1,...,15
$$
$$
\varepsilon_i \sim \mathcal{N}(0,7^2)\;,
$$

where $f_k$ are different nonlinear functions. Refer to table \ref{apB:functions_high_dim} for details. We then let the number of inputs that generate the true function increase (2,5,15) and for each generated function we add 1,10 and 20 noise terms. We estimate a single hidden layer ANN with 200 ReLU units and a local linear regression (kernel-type) on the data. To avoid having to choose the optimal number of sieve terms $r_n$, we use $L_1$-regularization in the hidden layer. For this and all subsequent simulations we only optimize the regularization parameters once and leave them unchanged throughout the simulations. We split the generated data into a 80/20 training-test set and calculate the RMSPE on the test data. The results are summarized in table \ref{table:high_dimensional_comparison}. 
%
\begin{table}[!htbp]
	\centering
	\begin{tabular}{@{\extracolsep{5pt}} ccccc} 
		\cline{3-5} \\[-1.6ex] 
		& & \multicolumn{3}{ c }{Relevant Predictors} \\ 
		\hline \\[-1.8ex] 
		& Noise & 2 & 5 & 15 \\ 
		\hline \\[-1.8ex] 
		\multirow{3}{*}{Nonparametric (LL)}
		& 1 &  $6.11^*$ & $18.31^*$ & $59.13$ \\ 
		& 10 & $28.98$ & $46.48$ & $79.30$ \\ 
		& 20 & $44.02$ & $66.65$ & $97.99$ \\  \hline \\[-1.8ex] 
		\multirow{3}{*}{ANN}
		& 1 & $22.33$ & $33.99$ & $41.86^*$ \\ 
		& 10 & $27.10^*$ & $36.85^*$ & $47.03^*$ \\ 
		& 20 & $34.87^*$ & $43.27^*$ & $51.72^*$ \\ \hline
		\\[-1.8ex] 
		\\[-1.8ex] 
	\end{tabular} 
	\caption[Performance of estimators across dimensions]{Comparison of the RMSPE across different estimators and data generating processes. The $^*$ denotes the best performing (in terms of RMSPE) model for the corresponding problem. For the first column of the ANN simulation we increased the learning rate slightly to make sure the loss levelled off. We also allowed for more training epochs for larger models. \label{table:high_dimensional_comparison} } 
\end{table} 

As expected, the kernel-type estimation performs really well when the function is generated by few terms and when few noise-terms are present. As the number of relevant predictors increases the ANN starts to perform (relatively) better. The lower sensitivity of the ANN against higher dimensions can also be seen with the increasing noise terms, where the RMSPE increase is relatively lower for the ANN. 

\subsection{Inference and Prediction Accuracy}\label{sim:inference}

A further advantage of the proposed partially linear model is that it allows clear inference on the parametrically specified part of the model which is an often sought-after property in financial modellin. By using the two step estimation procedure for the parametric part that we developed in section \ref{sec:formal}, we are able to obtain standard errors and a convergence to a distribution for our parametric estimates which is not possible in standard ANN estimation.

We will consider data generated from two different time series settings. One with, and another without temporal autoregression. For ease of interpretation and calculation, we will consider lower-dimensional series. We saw in the preceding section \ref{sim:high_dim} that a kernel-type estimator would probably provide more accurate results for such a problem, nevertheless, the present section allows us to illustrate the theoretical findings. Further, to the best of our knowledge, no general inference framework was yet proposed for semi-nonparametric ANN sieve estimators. The present section provides motivation for a development of such theory. We note that we allow all models the correct number of lags, as the focus of the analysis is not model selection but rather model comparison.

The first simulation is essentially a replication of example 11.11 in \cite{martin2012econometric}, who use it to illustrate the properties of a kernel estimator. We consider DGP of the form: 

$$
y_t = 2x_{1,t} + x_{2,t} + g(x_{3,t}) + u_t, \quad t=1,...,1250,
$$
$$
g(x_{3,t})=0.3\exp(-4(x_{3,t}+ 1)^2) + 0.7\exp(-16(x_{3,t}-1)^2),
$$
$$
u_t \sim \mathcal{N}(0,0.1), \quad x_{1,t}|x_{3,t} \sim \mathcal{N}(0.5x_{3,t}, 1),
\quad x_{2,t} \sim \text{student}_{df=4}, \quad x_{3,t} \sim \text{Unif}[-2,2]
$$

The parameters are then estimated with our proposed SANN from equation \eqref{eq:semi_multi_ann}, a kernel-based partially linear model, a benchmark linear model of the form $y_t=\beta_0 + \beta_1 x_{1,t}+\beta_2 x_{2,t} + \beta_3 x_{3,t}$ and a correctly specified model. To avoid having to choose the right sieve order for the SANN, we use $L_1$-regularization in the first hidden layer of the ANN. The prediction results, based on the RMSPE and the integrated bias and variance for the nonparametric part are reported in table \ref{table:prediction_int_semiparametric_model} (Model 1) and the nonparametric approximation is visualized in figure \ref{fig:fit_np_kernel_vs_sann}. From the results we can see that the kernel-type estimator outperforms the nonparametric approximation of the SANN, which is to be expected in such a low-dimensional problem. Nevertheless, the predictive performance of the SANN looks promising.

\begin{figure}
	\centering
	\includegraphics[width=0.45\linewidth]{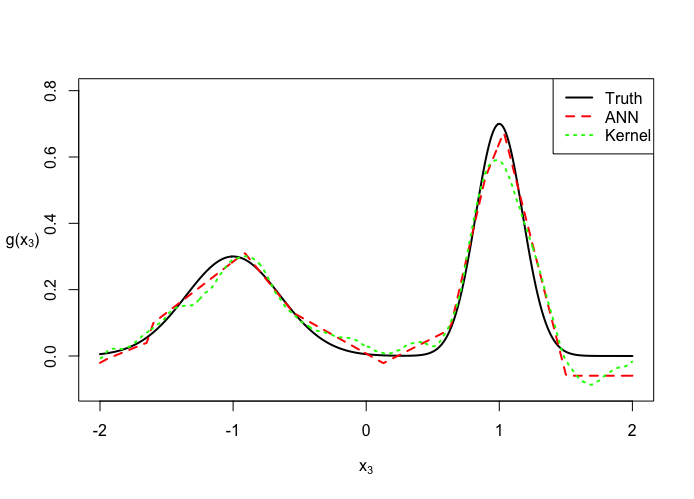}
	\caption[Nonparametric function approximation kernel vs. ANN-sieve]{Nonparametric function approximation for $x_3$ in Model 1.}
	\label{fig:fit_np_kernel_vs_sann}
\end{figure}

Next, we want to analyse the performance of the estimator in an autoregressive process time-series setting. For that, we consider the nonlinear multiple time series (which is closely related to the experiment of \cite{gao_nonlinear_timeseries}):

$$
y_t = \beta_1 v_{t-1} - \beta_2 v_{t-2} + \Big(\frac{x_{t-1} + x_{t-2}}{1 + x_{t-1}^2 + x_{t-2}^2}\Big)^2 + \varepsilon_t, \quad t=1,...,1250
$$
$$
\beta_1 = 0.47, \beta_2=-0.45
$$
$$
v_t = 0.55v_{t-1} - 0.42v_{t-1} + \delta_t, \quad x_t = 0.8\sin(2 \pi x_{t-1}) - 0.2\cos(2 \pi x_{t-2}) + \eta_t
$$
$$
\varepsilon_t \sim \mathcal{N}(0,0.5^2), \quad \delta_t,\eta_t \sim \text{Unif}[-0.5,0.5]
$$

We then run four different models on the data. The benchmark linear model, a hybrid ANN model, a fully nonparametric ANN model and the correctly specified model. Again, we use $L_1$-regularization to shrink the sieve order $r_n$ to the optimal level. Results are summarized in table \ref{table:prediction_int_semiparametric_model} (Model 2). We also present the distribution of $\beta_1$ and its standard errors in figure \ref{fig:distribution_parametric_estimate} ($\beta_2$ yielded similar results). Because the two series $v_t$ and $x_t$ are not correlated, the OLS prediction should provide unbiased and consistent estimates. We can infer that the SANN also provides accurate parameter estimates in addition to the increased predictive performance. Further, we can see that the standard errors from the SANN are overall lower than for the OLS estimation. This probably arises due to the efficiency gain from removing the nonparametric variance before estimating the parameters.

\begin{table}
	\centering
		\begin{tabular}{@{\extracolsep{5pt}} cccccc} 
			\\[-1.8ex]\hline \\[-1.8ex] 
			& Model 1 & Model 2 & Model 2a&$\int \text{Bias}^2$ & $\int \text{Var}_e$  \\
			\hline \\[-1.8ex] 
			True & $0.2482$ & $0.2450$ & & & \\
			Linear & $0.3647$ & $ 0.2714$ & & $ $ & $ $ \\ 
			ANN & $0.3434$ & $0.3034$ & & $ $ \\ 
			SANN & $0.3206$ & $0.2576$& $0.2609$ &$0.0012$ & $0.0072$ \\ 
			Kernel & $ $ & & &$0.0012$ & $0.0061$ \\ 
			\hline \\[-1.8ex] 
		\end{tabular} 
	 \caption[Performance of different semi-nonparametric models]{Performance statistics for semi-nonparametric models, based on the RMSPE for the 250 last observations (which were not used in the estimation process). For model 1 we also present the integrated squared bias and integrated variance of the SANN and Kernel-partially linear model. The performance metric is the RMSPE, it is not reported for the Kernel-PLM because the package does not support predictions. \label{table:prediction_int_semiparametric_model}}
\end{table}


Our simulation results also confirm what other empirical studies for time series have found before. The simple ANN provides worse results than the baseline linear estimate, even though we used regularization. On the other hand, our SANN which used the same amount of regularization and free parameters, improved our forecasting result. The improvement is not as large as in the previous experiment, which is probably due to the lower signal-to-noise ratio (also compare this to the correctly specified model, which only provides a $10.7\%$ improvement on the baseline linear model). Of course, it seems reasonable that a model which has the correct division of the parametric and nonparametric part, will outperform a purely linear or nonparametric model. Nevertheless, the results are motivating for the further use. We also estimated our SANN by providing all inputs to the linear and nonlinear part. Interestingly, the model still outperforms the baseline linear and ANN models (Model 2a) but, unsurprisingly, the performance is lower than that of the correctly specified SANN. 

\begin{figure}
	\begin{minipage}{0.5\linewidth}
		\centering
		{\includegraphics[width=0.9\linewidth]{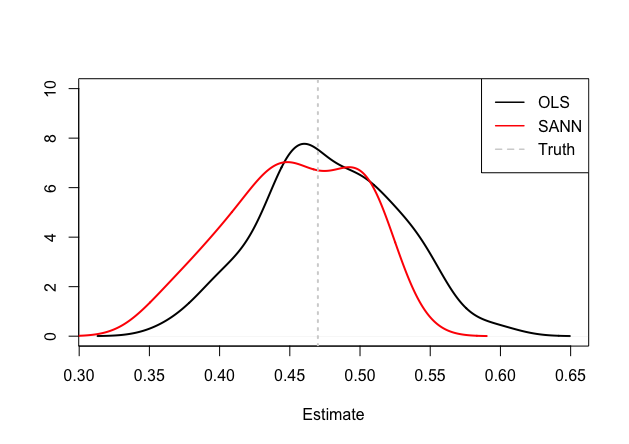}}
	\end{minipage}\hfill
	\begin{minipage}{0.5\linewidth}
		\centering
		{\includegraphics[width=0.9\linewidth]{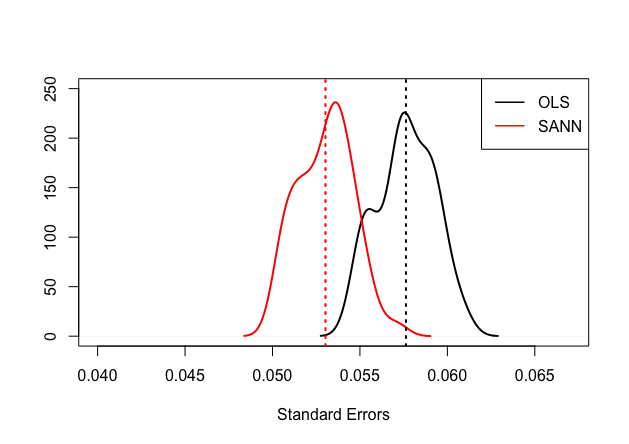}}
	\end{minipage}
		\caption[Distribution of Parametric Pstimate]{Distribution of estimates of the two parametric components in the stochastic time series simulation. Because there is no correlation between the parametric and nonparametric parts, a linear model should provide accurate estimates. \label{fig:distribution_parametric_estimate}}
\end{figure}

We conclude this section with a brief summary of the key findings. First, in higher dimensional settings, the ANN-sieve outperforms standard kernel regression. Second, when correctly specified, the SANN provides competitive results for function approximation and increases the efficiency of the parametric estimates compared to the linear baseline. Third, the SANN approach also seems to perform better than a pure ANN modelling approach, in line with what was found in \cite{makridakis2020m4}.

\section{Nonparameteric Portfolio VaR}\label{sec:portfolio}

The previous sections have derived the theoretical convergence and presented evidence on these properties with finite samples. The present section demonstrate its practical use and can be seen as a guideline on when sieve estimation via ANNs might be useful. 

We began the introduction with the argument that economic and financial time series often seem to have nonlinear components but that financial forecasting is often situated in high dimensional settings, as the diversification of portfolios involves many assets. This makes nonparametric estimation at the same time attractive, as it offers a way to avoid parametric misspecification but also unattractive because nonparametric estimators suffer from the curse of dimensionality. The last section illustrated, that the SANN can certainly model nonlinearities and might also perform better than other nonparametric approaches in higher dimensions.  

An important measure in statistical risk management is the Value at Risk (VaR). For a confidence level $\alpha \in (0,1)$ ($\text{VaR}_{\alpha}$) at horizon $h$, it is defined as the smallest value $p$ such that the probability of a loss ($L$) smaller than $p$ is at least $(1-\alpha)$. The definition corresponds to the $1-\alpha$ quantile of the loss distribution at time $t$ for the horizon $h$:

$$
\text{VaR}_{h,\alpha}=\inf \{p: \mathbb{P}(L_{t,t+h} \leq p)\geq (1-\alpha)\}
$$

For our results, we will set the forecast horizon $h$ to 1. The VaR$_\alpha$ is, for example, used to partially determine the reserves that a financial institution is required to hold. This opens up two aspects that are important from the modelling perspective. A well constructed estimate should reflect the risk of a loss greater than VaR$_\alpha$ accurately, but at the same time the estimate should also be as low as possible to avoid excessive reserves. 

\subsection{Semiparametric CAViaR}

A variety of estimation procedures exist for the estimation of the VaR. A specially interesting basic estimate is the CAViaR model from \cite{caviar_paper}. Instead of modelling the whole distribution of the return (and hence the loss) and then estimating the $(1-\alpha)$-quantile, the CAViaR approach attempts to model the VaR directly as an autoregressive function. The rationale behind this is the empirical finding that variances are autoregressive conditionally heteroskedastic. The VaR which is linked to the standard deviation of returns should hence also be conditionally linked to its previous observations. To specify a CAViaR model, let $y_t$ be a vector of returns, $\alpha$ the probability as described above, $x_t$ a vector of observable variables such as the lagged variance and $f_t(\theta) \equiv f_t(x_{t-1}, \theta_{\alpha})$ the $\alpha$-quantile of the distribution of portfolio returns at time $t$, parametrized by $\theta$. Following our notation from the last section, we propose to model the VaR as:

\begin{equation}\label{eq:semiparametric_caviar}
	f_t(\theta)= \sum_{j=1}^{p}\beta_j f_{t-j}(\theta) + \phi (x_{t-1})\; ,
\end{equation}

where $\phi(x_t)$ is a nonparametric function of lagged values, such as the variance, in the information set at $t-1$. Here, besides the predictive performance, we are interested in how much information can be drawn from past VaR observations. A significant coefficient would suggest that clustering of volatilities even has effects in the tails of the distribution (\cite{caviar_paper}). 

The model is simply an augmented version of the original CAViaR-GARCH model proposed in the original paper. Specified as a partially linear model, we can give some meaning to the VaR. It can be considered an autoregressive process driven by past VaR, with a time-varying intercept driven by an unknown process of the lagged variance. Because the aim of this section is the illustration of possible applications for the proposed estimator and not the discussion of statistical properties of the CAViaR model itself, we do not further engage in a theoretical discussion and refer the interested reader to \cite{caviar_paper}.

Because we are modelling a quantile, we need to adapt the optimization problem from equation \eqref{eq:sieve_estimator} to:
$$
\hat{\theta} = \arg \min_{\theta \in \Theta_n} \frac{1}{T}\sum_{t=1}^{T}[\alpha - \mathbbm{1}(y_t < f_t(\theta))][y_t-f_t(\theta)],
$$

where $\alpha$ is the quantile as defined before. Also note that we are no longer in a standard feed forward case, because the autoregressive VaR$_\alpha$-term is itself an estimation and needs to be fed back into the model. The autoregressive term is initalized as the unconditional $\alpha$-quantile of the training data. We let the lagged quantile enter linearly into our model, because even in such a case, the lagged value will contain the nonlinear transformations of the lagged variables in $x_{t-1}$. 

\begin{figure}
	\centering
	\includegraphics[width=0.7\linewidth]{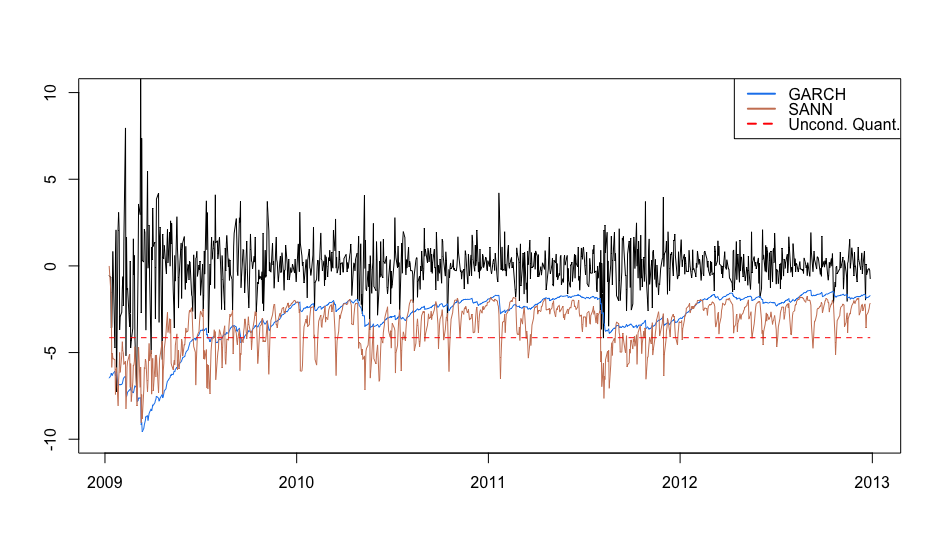}
	\caption[Univariate Semiparametric VaR estimates]{Out of sample prediction of the semiparametric VaR estimates from the SANN}
	\label{fig:testplotvarforecast}
\end{figure}

We estimate the model in equation \eqref{eq:semiparametric_caviar} for $p=0$ and $p=1$ with our SANN and compare its performance to that of a standard GARCH(1,1) model as well as the symmetric absolute value CAViaR. The data considered is the closing price of the GE-stock from 01.01.2000-31.12.2012. Deliberately, we include the stress period of the great recession and the European debt crisis, because we expect stress periods to have return distributions which are not in line with the normality assumption. We leave out the last 1000 observations for out of sample forecasting.

We then evaluate the forecasts with two statistical tests based on \cite{christoffersen_comparingVaR} and \cite{christoffersen_durationVaR}. First we check whether the number of exceedances of the VaR is in line with the expected number (which for $\alpha=0.01$ is 10). Further, if the model is correctly specified, the exceedances should be independent (otherwise there is still information left that could be incorporated into the model). A correctly specified model would fail to reject the null, that the model is correctly specified, in both cases. To assess the economic perspective, we numerically integrate the VaR estimates over the entire sample period and compare the percentage change of the value to that of the unconditional $\alpha$-quantile of the sample. Ideally, a good model would decrease the integrated VaR vis-a-vis the unconditional quantile. Finally, we report the estimate and standard errors for the autoregressive VaR$_\alpha$-term to check whether conditional heteroskedasticity is also relevant in the tails. The results are summarised in table \ref{table:univariate_caviar_results}.
%

\begin{table}[!htbp] \centering 
	\begin{tabular}{@{\extracolsep{5pt}} ccccc} 
		\hline \\[-1.8ex] 
		& GARCH(2,1) & SANN & SANN (recursive) & CAViaR \\ 
		\hline \\[-1.8ex] 
		Exceedances & $14$ & $11$ & $12$ & $14$ \\ 
		Failures Test & $0.23$ & $0.75$ & $0.54$ & $0.23$ \\ 
		Duration Test  & $0.15$ & $0.01$ & $0.17$ & $0.42$ \\ 
		Change $\int$ VaR$_\alpha$ & $-27.7\%$ & $-13.8\%$ & $-15.8\%$  & $-18.7\%$\\ 
		Coef. VaR$_{t-1}$ & $$ & $$ & $0.66$ & $0.82$ \\ 
		Std. Error & $$ & $$ & $0.19$ &  $$ \\ 
		\hline \\[-1.8ex] 
	\end{tabular} 
	\caption[Results of univariate VaR$_{0.01}$ models]{Results of univariate VaR$_{0.01}$ forecasting models. The recursive SANN allowed for a lag in the dependent variable. For the failures and duration tests the p-values are reported. The change towards $\int$VaR$_\alpha$ is compared to the integrated unconditional quantile. The VaR$_{t-1}$ coefficient and standard error are from the in-sample estimation.  	\label{table:univariate_caviar_results}} 
\end{table}

For the univariate case, all the models provide VaR-estimates that are in line with a correct specification once we allow for an autoregressive ($p=1$) term in the SANN. Both tests on the correct number of exceedances (failures) and the independence of the exceedances (duration) are not rejected, although the SANN has fewer exceedances. The parameter estimate and its significance are also in line with what was found in \cite{caviar_paper} and are similar for both the CAViaR and SANN. Yet, based on the statistical evidence, we cannot conclude that the SANN improves the results for the given stock and period. The economic gain from using the SANN to model the VaR is also smaller than that of the GARCH model, as the decrease of the integrated VaR is smaller. Further, as can be seen from the out of sample one-step ahead predictions in Figure \ref{fig:testplotvarforecast} the SANN seems to overfit the data to some extent (as seen by the unsmooth VaR estimate). For this specific problem, it does not seem that the use of (the more complicated) SANN is justified. 

\subsection{Comparison in Portfolio VaR}

Given the flexibility of the SANN structure, the model can easily incorporate a multivariate structure or the use of different correlated series. In the second part of this application, we will model the VaR of a portfolio instead of a single stock. Whereas the use of the SANN on a single series resulted in medicore results, the use in a higher dimensional setting with potentially nonlinear inter-series relationships might bear better results. We construct 50 randomly weighted portfolios of six assets, that are closer described in appendix \ref{apB:application_specification}.

\begin{table}[!htbp] \centering 
	\begin{tabular}{@{\extracolsep{5pt}} cccc} 
		\hline \\[-1.8ex] 
		& MGARCH(1,1) & SANN & CAViaR \\ 
		\hline \\[-1.8ex] 
		Exceedances & $23.02$ & $10.56$ & $22.20$\\ 
		Failures Test & $0.52$ & $0.02$ & $0.50$ \\ 
		Duration Test &  $0.12$ & $0.06$ & $0.16$\\ 
		Change $\int$ VaR$_\alpha$ & $-25.0\%$ & $-6.1\%$ & $-8.7\%$\\ 
		Coef. VaR$_{t-1}$ & & $0.46$ & $0.55$\\
		Std. Error & & 0.11 & \\
		\hline \\[-1.8ex] 
	\end{tabular} 
	\caption[Results of multivariate VaR$_\alpha$ models]{Results of VaR$_{0.01}$ out of sample forecasts on the portfolio level. The statistical tests were calculated for the 0.95 level and the figures represent the average rejection rate over the 50 constructed portfolios. \label{table:results_multivariate_application}} 
\end{table}

We compare the performance across the 50 portfolios and summarize the results in table \ref{table:results_multivariate_application}. Contrary to the univariate case, the MGARCH(1,1) model fails to provide accurate VaR$_{0.01}$ forecasts for the portfolios in 52\% of the cases (note that the table does \emph{not} report the p-values but the averages of the rejection rate), The performance for the CAViaR is similar to that of the MGARCH model. The SANN on the other hand provides very accurate results for the portfolio VaR. The models do not appear to have significant improvement possibilities (as all do no reject the duration test that would detect residual autoregressive information). As in the univariate case, the average decrease in the integrated VaR compared to the unconditional baseline is smaller than that of the MGARCH model and to some extent the CAViaR model. But this might reflect just the true risk better, especially in turbulent markets, where the MGARCH and CAViaR models frequently underestimated the Value at Risk. The findings themselves would indicated that there are some nonlinear relationships between the assets that are not properly account for with the more traditional models. The SANN can take this into account and at the same time is less sensitive to the dimension of the portfolio compared to other nonparametric estimators.

\section{Conclusion}\label{sec:conc}

The goal of the article was to illustrate in which circumstances the use of ANNs might bring about an increase in predictive performance with a focus on financial time-series. By incorporating ANNs into the framework of sieve estimation based on \cite{handbook_econometrics_chen} and extending the theoretical results of \cite{chen_racine_2001} we were able to derive the rate of convergence for ReLU based ANNs . It turns out, that the use of ReLU activation functions, allows to achieve the same rate of convergence as activation functions that suffer from the vanishing gradient problem, which further motivates their use. In light of recent studies, such as \cite{makridakis2020m4}, a general semiparametric model was then proposed and its bias variance trade-off and convergence properties analysed on a statistical, rather than a heuristic basis. With extensive monte-carlo simulations we provided evidence that ANN-based nonparametric estimates tend to converge better in higher dimensions than other nonparametric methods, such as kernel regression. Furthermore, inference on the parametric part of the SANN is also possible and provided promising results. Finally, the proposed SANN estimator performed better in out of sample forecasts than both the standard linear model or a fully nonparametric ANN by themselves. These results also held when we applied the model to real stock market data to estimate the Value at Risk of a portfolio. The findings are generally in line with the literature but especially so with what was found in \cite{makridakis2020m4}. ANNs can effectively exploit (nonlinear) relationships between different time-series and, as was shown in section \ref{sec:formal} theoretically and sections \ref{sec:finite} and \ref{sec:portfolio} with finite samples, seem to be less affected by high-dimensional problems than other nonparametric estimators.

In the statistical analysis of ANNs, much remains to be done. For example, the applications simply involved a basic linear component in the parametric part. But this approach could be extended to include richer parametric structures and error corrections. Since ANNs seem to perform better than other nonparametric techniques in higher dimensional settings, they are particularly attractive for portfolio modelling. Future research can focus on the many potential scenarios have not yet been explored, for example how ANNs behave in the presence of cointegration. Moreover, long-term dependencies arise frequently in financial or economic settings. Throughout this article we only used the simplest form of ANNs, but there exists a variety of ANNs that allow for recurrence relations in the data and could be further investigated from a statistical perspective

\bibliography{bibliography}

\section{Appendix}\label{sec:app}

\vspace{1cm}
\subsection{Derivation Bias-Variance Tradeoff} \label{apA:bias_variance}

We will derive the bias-variance tradeoff for ANNs with the methodology of \cite{hansen_2014_hb}.

\begin{proof}[Proof of proposition \ref{theorem:bias_variance}]
	Without loss of generality we will derive the mean integrated squared error for the sieve order $r_n$ and data set size $n$ ($\text{MISE}_n(r_n)$). The generalisation for stationary time series is then straightforward. Define the MISE as:
	
	$$
	MISE_n(x)= \mathbb{E} \int \big (\hat{\psi}_m(x) - \psi(x) \big )^2dx
	$$
	
	Contrary to the standard OLS case, we need to account for two sources of error: the approximation error and the regression error. From the definition of the approximation error in \ref{eq:definition_approximation_error}, we know that $\psi(x)=\upsilon_n + \mathcal{G}(x)\beta_n$. We can then rewrite the RHS of the equation as\footnote{For the ease of notation we omit the time index $t$ from the calculations. In this case assume that $\mathcal{G}(x)$ is the $(1 \times r_n)$ vector of transformed variables for $x_t$ in the first hidden layer and $u$ the corresponding noise term.}:

	\begin{equation*}
	\begin{split}
		\int \big (\hat{\psi}_n(x) - \psi(x) \big )^2f(x)dx & = \int \Big (\mathcal{G}(x)\hat{\beta}_n - (\mathcal{G}(x)\beta_n +\upsilon_{n}(x))\Big)^2 \\
		& = \int \upsilon_n(x)^2f(x)dx - 2(\hat{\beta}_n - \beta_n)'\int \mathcal{G}(x)'\upsilon_nf(x)dx \\
		& + (\hat{\beta}_n - \beta_n)' \int \mathcal{G}(x)'\mathcal{G}(x)f(x)dx (\hat{\beta}_n - \beta_n) \\
	\end{split}
	\end{equation*}
	
	We then use  $\mathbb{E}[\mathcal{G}(x)'\upsilon_n] =0$ (from the linear regression assumption) to let the second term be equal to 0. We use the definition of $\varphi_{n}^2$  and set $\mathbb{E}[\mathcal{G}(x)'\mathcal{G}(x)] \equiv \mathcal{E}_n$. Further by the property of $\mathbb{E}[x]=\mathbb{E}[trace(x)]$ for $x$ a scalar and the cyclical property of the trace we get:
	
	$$
	IMSE_n(x)  = \varphi_{n}^2 + trace\Bigg (\mathcal{E}_n \mathbb{E}\Big [(\hat{\beta}_n - \beta_n) (\hat{\beta}_n - \beta_n)' \Big ] \Bigg)
	$$
	
	Finally, define $\mathbb{E}[\mathcal{G}(x)'\mathcal{G}(x)\sigma_t^2] \equiv \Omega_n$ where $\sigma_t^2 = \mathbb{E}[u_t^2|x_t]$. Using the formula of the least squares estimate for $\hat{\beta}_n$ and the true model of the sieve estimate of order $n$, $y_t= \mathcal{G}(x_t)\beta_n + u_t$ we get the following (asymptotic) result:
	
	\begin{equation*}
		\begin{split}
			\mathbb{E}\Big [(\hat{\beta}_n - \beta_n) (\hat{\beta}_n - \beta_n)' \Big ] & = \mathbb{E} \Big [ \big ( (\mathcal{G}(x)'\mathcal{G}(x))^{-1}(\mathcal{G}(x)'\mathcal{G}(x)\beta_n) + (\mathcal{G}(x)'\mathcal{G}(x))^{-1}(\mathcal{G}(x)'u) - \beta_n \big) \\
			& \;\;\;\;\;\;\;\;\;  \big ( (\mathcal{G}(x)'\mathcal{G}(x))^{-1}(\mathcal{G}(x)'\mathcal{G}(x)\beta_n) + (\mathcal{G}(x)'\mathcal{G}(x))^{-1}(\mathcal{G}(x)'u) - \beta_n \big)'
			 \Big ] \\
			 & = \mathbb{E} \Big [ (\mathcal{G}(x)'\mathcal{G}(x))^{-1}\mathcal{G}(x)'uu'\mathcal{G}(x)  (\mathcal{G}(x)'\mathcal{G}(x))^{-1} \Big ] \\
			 & = n^{-1} \mathcal{E}_n^{-1}\Omega_n \mathcal{E}_n^{-1}
		\end{split}
	\end{equation*}
	
	Note that we used the sample equivalent in the last step. Together this yields:
	$$
	MISE_n(x)  = \varphi_{n}^2 + n^{-1}trace(\mathcal{E}_n^{-1} \Omega_n)
	$$
	
	Which is proposition \ref{theorem:bias_variance} for sieve order $r_n$.
	
\end{proof}

\subsection{Parametric Part of a Partially Linear Model}\label{apA:parametric_part}

\begin{proof}

Following our notation from above we denote $\mathcal{G}(Z)$ nonlinear transformed, $(n \times r_n)$ output matrix in the last layer and $X$ the $(n \times k)$ linear inputs. The optimization problem in the output layer is then similar to the standard OLS problem and we can then write it the equation as:

$$
A= 
\begin{bmatrix}
X & \mathcal{G}(Z)
\end{bmatrix}
,
\Lambda=
\begin{bmatrix}
\beta_1 \\
\beta_2
\end{bmatrix}
$$

\begin{equation}\label{eq:last_layer_equation}
	Y =  X\beta_1 + \mathcal{G}(Z)\beta_2 + \varepsilon = A \Lambda + \varepsilon
\end{equation}

Given the OLS solution to equation \ref{eq:last_layer_equation} we then transform the problem into the normal equations:

$$
\hat{\Lambda}_{OLS} = (A'A)^{-1}A'Y \implies (A'A)\hat{\Lambda}_{OLS}=A'Y
$$
Which by the above definition can be rewritten as:
$$
\begin{bmatrix}
	X'X & X'\mathcal{G}(Z) \\
	\mathcal{G}(Z)'X & \mathcal{G}(Z)'\mathcal{G}(Z)
\end{bmatrix}
\begin{bmatrix}
	\hat{\beta}_1 \\
	\hat{\beta}_2
\end{bmatrix}
=
\begin{bmatrix}
	X'Y \\
	\mathcal{G}(Z)'Y
\end{bmatrix}
$$

Note that the matrix $A$ needs to be invertible, which means that both $X$ and $\mathcal{G}(Z)$ need to have full column rank. ANNs (especially with ReLU activation functions) tend to have at least some columns equal to 0 in the final layer, therefore they need to be removed prior to the estimation.We can then follow \cite{greene_econometric} theorem 3.2 to obtain the estimate for $\beta_1$\footnote{In the case where $X$ and $\mathcal{G}(Z)$ are orthogonal, we could directly regress $Y$ on the corresponding matrices to get the estimate of $\beta_1$ and $\beta_2$. However, we treat the general case, where we would for example allow the ANN to contain all variables or the case where the sieve order was chosen without data driven methods.}

$$
	\hat{\beta_1} = (X'M_{\mathcal{G}}X)^{-1}(X'M_{\mathcal{G}}Y)
$$
Where $M_{\mathcal{G}}$ is the ''residual maker matrix'' that contains the residuals from a regression of $Y$ on $\mathcal{G}(Z)$:

\begin{equation*}
\begin{split}
	Y - \mathcal{G}(Z)\hat{\beta}_2 & = Y -  \mathcal{G}(Z)( \mathcal{G}(Z)' \mathcal{G}(Z))^{-1} \mathcal{G}(Z)'Y \\
	& = [I_n -  \mathcal{G}(Z)( \mathcal{G}(Z)' \mathcal{G}(Z))^{-1} \mathcal{G}(Z)']Y \\
	& = M_{\mathcal{G}}Y
\end{split}
\end{equation*}

Last, note that with the theory developed about ANN-Sieves we know that $\mathcal{G}(Z)$ will converge to the true function of $g(Z)$ and hence we can calculate the approximation of $\tilde{Y}=Y-\mathbb{E}[X|Z]$ and $\tilde{X}=X - \mathbb{E}[X|Z]$. As we only approximate the true function $g$ and hence $\tilde{X}$ and $\tilde{Y}$ we denote the approximations with $\tilde{X}^{*}=X-\hat{\mathbb{E}}[X|Z]$ and $\tilde{Y}^{*}=Y-\hat{\mathbb{E}}[Y|Z]$ respectively. Hence we get:

\begin{equation*}
	\hat{\beta}= (X'M_{\mathcal{G}}X)^{-1}(X'M_{\mathcal{G}}Y) = (\tilde{X}^{*'}\tilde{X}^{*})^{-1}\tilde{X}^{*'}\tilde{X}^{*}
\end{equation*}

By using the fact that $M_{\mathcal{G}}$ is symmetric and idempotent.

\end{proof}

\subsection{Proof of Proposition \ref{proposition:convergence_relu_sieves}} \label{apA:proof_convergence_sieves}
We follow a similar way as \cite{chen_racine_2001} and \cite{handbook_econometrics_chen} who prove the proposition for ANNs with smooth sigmoid and gaussian radial basis activation functions. We follow their methodology  state the following: In order for the sieve-estimate $\hat{\theta}_n$ to converge to the true value $\theta_0$ we need to minimize the approximation error $||\theta_0 - \pi_n\theta_0||$ to avoid asymptotic bias (because the sieve spaces are only in the limit equal to the true function space), but at the same time the sieve space should not be too complex to avoid overfitting. To measure the complexity of the sieve space we use the $L_2$ metric entropy with bracketing (see eg. \cite{chen_shen_1998} for a definition) of a class $\mathcal{F}_n=\{g(\theta, \cdot):\theta \in \Theta_n\}$. We denote this by $H_{[]}(\omega, \mathcal{F}_n, ||\cdot||_2)$. We can then state the same conditions for the convergence of sieve M-estimators as in \cite{chen_shen_1998}:

\begin{condition}\label{condition:beta_mixing}
	${Y}_{t=1}^n$ is a stationary $\beta$-mixing sequence, with $\beta_\tau \leq \beta_0\tau^{-\xi}$ for some $\beta_0 > 0, \xi = \gamma - 2 > 0$ (see definition of $\gamma$ below)
\end{condition}

\begin{condition}\label{condition:varepsilon_1}
	$\exists \; c_1 > 0$ such that for small $\varepsilon > 0$,
	$$
	\sup_{{\theta \in \Theta_n: ||\theta_0 - \theta|| \leq \varepsilon}} \textnormal{Var}\big(l(\theta, Y_t)- l(\theta_0, Y_t)\big) \leq c_1\varepsilon^2
	$$
\end{condition}

\begin{condition}\label{condition:delta_value}
	Let $\mathcal{F}_n=\{l(\theta, Y_t)-l(\theta_0, Y_t): ||\theta_0, \theta||_2 \leq \delta, \theta \in \Theta_n\} \; \exists \; \delta_n \in (0,1)$ and constants $c_2$ and $b>0 $ such that:
	$$
	\delta_n =  \inf  \bigg\{\delta \in (0,1): \frac{1}{\sqrt{n}\delta^2}\int_{b\delta^2}^{\delta} \sqrt{H_{[ ]}(\omega, \mathcal{F}_n, ||\cdot||_2)} d \omega \leq c_2  \bigg \} 
	$$
\end{condition}

\begin{condition}\label{condition:delta_2}
	For any $\delta > 0,  \exists \; s \in (0,2)$ such that
	$$
	\sup_{{\theta \in \Theta_n: ||\theta_0 - \theta|| \leq \delta}} |l(\theta, Y_t)- l(\theta_0, Y_t)| \leq \delta^sU(Y_t)
	$$
	for some function $U$ with $\mathbb{E}[(U(Y_t))^{\gamma}]$ for some $\gamma > 2$
\end{condition}

The conditions ensure that within a neighbourhood of $\theta_0$, $l(\theta, Y_t)$ is continuous at $\theta_0$ and that $||\theta_0 - \theta||^2$ behaves locally as the variance (cf. \cite{chen_shen_1998}). \cite{handbook_econometrics_chen} then gives the following theorem:

\begin{theorem}[\cite{handbook_econometrics_chen} Theorem 3.2]
	Let $\hat{\theta}_n$ be the approximate sieve M-estimator defined by \eqref{eq:approximate_sieve_est} and \eqref{eq:sieve_m_formula}, then under conditions 1-4 we get:
	$$
	||\theta - \hat{\theta}_n|| = \mathcal{O}(\max\{\delta_n, ||\theta_0 - \pi_n \theta_0||\})
	$$
	\label{theorem:convergence_sieve_ann}
\end{theorem}

Which will be useful in the proof.

\begin{proof}[Proof of proposition \ref{proposition:convergence_relu_sieves}]
	
	First, in order to use existing results, we note that by \cite{white_stinchcombe_weights} the ReLU function can be made to comply with the assumption of compact support of the activation function. For the theoretical considerations it is sufficient to think about a constant $c_l > 0$ such that we consider $f(x)=\min\{\max\{0,x,\}, c_l\}$. Which can be achieved by restricting the weights $\gamma_{i,j}$ suitably (note that under assumption \ref{assumption:bounded_beta} and for the single hidden layer model this is guaranteed). Under these assumptions, the ReLU-ANN sieve is (in the limit) dense in $\mathcal{P}$ by theorem \ref{theorem:universal_approx_theorem}.
	
	Throughout the proof assume that assumptions \ref{assumption:beta_mixing}, \ref{assumption:non_mix_between_npandp} and \ref{assumption:bounded_beta} hold. This implies that condition \ref{condition:beta_mixing} is directly satisfied.
	
	Further, conditions \ref{condition:varepsilon_1} and \ref{condition:delta_2} can be verified the same way as in \cite{handbook_econometrics_chen} example 3.2.2.
	
	To obtain the deterministic approximation error rate, note that we restricted the parameters and the maximum of our ReLU function, which means that the Hölder assumption in \cite{chen_white_1999} is fulfilled. Further it is clear that the ReLU function is \textit{not} homogenous (ie. $\text{ReLU}(\lambda x)\neq \lambda \text{ReLU}(x)$ in general). Then, by theorem 2.1, and lemmas A1, A2 from \cite{chen_white_1999} and the adjustment to our problem we obtain the deterministic approximation error rate :
	
	$$
	||\theta - \pi_n\theta|| \leq \textnormal{constant}*r_n^{\frac{1}{2}}\varepsilon_n(\mathcal{A})=\mathcal{O}(r_n^{-\frac{1}{2} - \frac{1}{d+1}})
	$$
	
	We note that our $constant$ in the equation above differs slightly from the one in \cite{chen_white_1999}, as we do not restrict the domain of our activation function to $(0,1)$ but to $[0,c_l]$. Nevertheless we obtain a constant, in this case assume that it is multiplied by $c_l$.
			
	What is left is to verify condition \ref{condition:delta_value}. To obtain the metric entropy with bracketing we employ \cite{shen_wong_1994} lemma 5 and the methodology of \cite{white_covering_ann} lemma 4.3. We summarise the results in the following claim:
	
	\begin{remark}[Bound on metric entropy with bracketing]
		An upper bound on the metric entropy with bracketing for sieve order $r_n$ and $\omega > 0$ can be expressed as:
		$$
			\mathcal{H}_B(\omega, \mathcal{F}_n, ||\cdot||_2) \leq \text{constant}*r_n\log(\frac{\text{constant}*r_n}{\omega})
		$$
	
	\end{remark}

		To prove the claim, we first note that by \cite{ossiander_1987} and \cite{handbook_econometrics_chen} (p. 5595) an upper bound on the bracketing metric entropy is sufficient and together with condition \ref{condition:delta_2}, we can write this upper bound as:
		
		$$
		\mathcal{H}_{[]}(\epsilon, \mathcal{F}_n, ||\cdot||) \leq \log N(\epsilon^{\small \frac{1}{s}}, \Theta_n, ||\cdot||) \leq \log N(\epsilon^{\small \frac{1}{2}}, \Theta_n, ||\cdot||)\;,
		$$
		
		 where $\mathcal{F}_n$ and $\Theta_n$ are as defined in condition \ref{condition:delta_value}. Note that we adapt the notation slightly by changing $\omega$ to $\epsilon$, to separate the claim from the rest of the proof. $N(\epsilon,\Theta_n, \rho)$ is the \textit{covering number} which is the cardinality of the smallest set of open balls with radius $\epsilon$ that is needed to cover (the set) $\Theta_n$. We denote the set of centers of these balls $T_\epsilon$ (which is also called an $\epsilon$-net). $\rho$ is the corresponding metric. This also illustrates the rationale behind condition \ref{condition:delta_value}. As the sieve space increases, the variance of the estimation increases. By controlling for $\delta_n$ in theorem \ref{theorem:convergence_sieve_ann}, we take that into account.
		 
		 For $T_\epsilon$ to be a valid $\epsilon$-net of $\Theta_n$, we must have that $\forall \; \theta \in \Theta_n \; \exists \; t_k \in T_\epsilon \text{ such that } \rho(\theta, t_k) < \epsilon$. We can then follow \cite{white_covering_ann} (proof of lemma 4.3), set $\eta>0$ and let $B_\eta \equiv \{b_k \in B, k=1,...,l\}$ and $G_\eta \equiv \{g_k \in G, k=1,...,q\}$ be $\eta$-nets for $B = \{\beta: ||\beta|| \leq c_n\} \; \subset \mathbb{R}^{r_n}$ and $G = \{\gamma: ||\gamma|| \leq r_nc_{n}\}\; \subset \mathbb{R}^{r_n(1 + d)}$ respectively (which we imposed with assumption \ref{assumption:bounded_beta}). We will also follow \cite{white_covering_ann} and employ the $L_1$-norm. Because we are in a finite dimensional space, it will hence also be an upper bound for, for example, the $L_2$-norm.
		 
		 Then, we denote $M_\eta=B_\eta \times G_\eta$ and let $\tilde{T}_\eta$ be the $\eta$-net for $\{\theta \text{ such that } \theta \in M_\eta \}$. We can then choose $\eta$ such that $\tilde{T}_\eta$ will be an $\epsilon$-net. By using our definition of $\theta$ from equation \eqref{eq:single_hidden_nn} ($\theta = (\beta', \gamma')'$) we have that for an arbitrary $\theta$ there will exist $\mu=(b', g')' \in M_\eta$ such that $||\beta - b||<\eta$ and $||\gamma - g|| < \eta$. Denote $t(\theta)$ the element of $\tilde{T}_\eta$ corresponding to $\mu$. We can then solve (again, analogous to \cite{white_covering_ann}):

		 \begin{equation*}
			 \begin{split}
			 |f_n(\tilde{x_t}, \theta)- f_n(\tilde{x_t}, t(\theta))|
			 & = |\sum_{i=1}^{r_n}\beta_i G(\tilde{x_t}\gamma_i) - \sum_{i=1}^{r_n}b_i G(\tilde{x_t}g_i) + \sum_{i=1}^{r_n}b_i G(\tilde{x_t}\gamma_i) -  \sum_{i=1}^{r_n}b_i G(\tilde{x_t}\gamma_i)| \\
			 & \leq |\sum_{i=1}^{r_n}(\beta_i-b_i) G(\tilde{x_t}\gamma_i)| + |\sum_{i=1}^{r_n}b_i(G(\tilde{x_t}\gamma_i)-G(\tilde{x_t}g_i))| \\
			 & \leq ||\beta - b||c_l + c_n(d+1)L||\gamma- g|| \\
			 & \text{where we note that our ReLU function is lipschitz continuous} \\
			 & \text{with lipschitz constant L=1 and we use standardized inputs x.} \\
			 & \leq \eta(c_l + c_n(d+1)) = \epsilon
			 \end{split}
		 \end{equation*}
		 
		 Which means we have to set $\eta=\epsilon/(c_l + c_n(d+1))$. We then have that $T_\epsilon \equiv \tilde{T}_{\epsilon/(c_l + c_n(d+1))}$ is an $\epsilon$-net for our ANN parameters. To calculate the covering number we can then use that the covering number of $\tilde{T}_\eta$ is bounded by the product of the covering number of its components $B_\eta$ and $G_\eta$. We can then employ lemma 5 from \cite{shen_wong_1994} for the $L_2$-norm, with our calculated $\eta$. Note that we adapt the $\delta$ to $c_n$ for $B$ and $c_nr_n$ for $G$ respectively (and transform $n$ term to $r_n$ and $r_n(d+1)$). For the $L_2$ norm, $N_1$ from lemma 5 \cite{shen_wong_1994} for $G$ can then be expressed as:
		 $$
		 N_1 \leq \frac{(\pi^{\frac{1}{2}} c_nr_n)^{r_n(d+1)}/\Gamma(\frac{c_nr_n}{2}+1)}{(\eta /\sqrt{r_n(d+1)})^{r_n(d+1)}}
		 $$
		 Which can then be solved by plugging in the calculated $\epsilon$ and approximating the solution with Stirling's formula. Finally, we note that only the sieve term $r_n$ and the radius of the n-ball $\epsilon$ are not fixed a priori. We can therefore summarize the other terms in the $constant$ and obtain:
		 
		 $$
		 	\mathcal{H}_{[]}(\epsilon, \mathcal{F}_n, ||\cdot||_2) \leq \text{const.}*r_n(d+1)\log(\frac{\text{const.}*r_nc_l(d+1)c_n}{\epsilon}) = \text{const.}*r_n\log(\frac{\text{const.}*r_n}{\epsilon})
		 $$
		 
		 Which is the same as claimed. The constant is positive and contains the $\frac{1}{s}$ term from \cite{handbook_econometrics_chen} (p. 5595). Note that this also corresponds to the bound on the entropy in \cite{chen_white_1999}, but the constant is multiplied additionally by $c_l$.

	Then in order to verify condition \ref{condition:delta_value}, we simplify the equation (where we leave out the $constant$ terms for ease of notation because they are positive and can hence be taken to the RHS of the inequality in condition \ref{condition:delta_value}):
	
	\begin{equation*}
	\begin{split}
		\frac{1}{\sqrt{n}\delta_n^2} \int_{b\delta_n^2}^{\delta_n}\sqrt{r_n\log(\frac{r_n}{\omega})}d\omega & \leq \frac{1}{\sqrt{n}\delta_n^2} \sqrt{r_n} \Big[ \big(\int_{b\delta_n^2}^{\delta_n} \log(\frac{r_n}{\omega})d\omega \big)^{\frac{1}{2}}(\delta_n-\delta_n^2b)^{\frac{1}{2}} \Big] \\
		& = \frac{1}{\sqrt{n}\delta_n^2}\sqrt{r_n}(\delta_n-\delta_n^2b)^{\frac{1}{2}} \Big[ (\delta_n-\delta_n^2b) \log(r_n) + \log(\delta_n^2b)\delta_n^2b  \\
		& - \log(\delta_n)\delta_n + (\delta_n-\delta_n^2b)\Big]^{\frac{1}{2}} \\
		& \leq \frac{1}{\sqrt{n}\delta_n^2}\sqrt{r_n}(\delta_n-\delta_n^2b)^{\frac{1}{2}} \Big[(\delta_n-\delta_n^2b)\log(r_n) - (\delta_n-\delta_n^2b)\log(\delta_n)\\
		& + (\delta_n-\delta_n^2b) \Big]^{\frac{1}{2}} \\
		& \bigg (\textnormal{By setting eg: } b= \frac{\log(\delta_n)-1}{\delta_n(\log(\delta_n) -1 -r_n)} \geq 0,  \forall r_n \geq 1 \bigg )\\
		& \leq \frac{1}{\sqrt{n}\delta_n^2}\sqrt{r_n}\sqrt{\log(r_n)}\delta_n
	\end{split}
	\end{equation*}
	
	Where we used the Cauchy-Schwarz theorem in the first inequality and the fact that $\delta_n \in (0,1)$. Note that we can choose the value for the constant $c_2$ to satisfy the condition \ref{condition:delta_value} (ie. $\delta_n \in (0,1)$). Rewriting the condition:
	$$
	\delta_n \geq c_2 \sqrt{\frac{r_n \log(r_n)}{n}}
	$$
	
	Finally, if we set $\delta_n=||\theta_0-\pi_n\theta_0||$ to balance bias and variance as suggested in \cite{chen_shen_1998}, we can establish a rate for the sieve order $r_n$:
	$$
	c_2 \sqrt{\frac{r_n \log(r_n)}{n}} \leq \delta_n=||\theta_0-\pi_n\theta_0|| \leq \text{constant}*r_n^{-(\frac{1}{2} + \frac{1}{d+1})}
	$$
	$$
	\implies r_n^{2(1 + \frac{1}{d+1})}\log(r_n) \leq \bigg( \frac{\text{constant}}{c_2} \bigg)^2 n = \mathcal{O}(n)
	$$
	
	The last part of the proof then consists of establishing the convergence of the sieve estimate based on $n$. It is sufficient to calculate the convergence of the dominant term in the relation above. Hence we can proceed in two steps:
	
	\begin{remark}[Dominant term]
		The dominant term of $r_n^{2(1+\frac{1}{d+1})}\log(r_n)$ is $r_n^{2(1+\frac{1}{d+1})}$
	\end{remark}

	To prove the claim we note the following, our dataset of size $n$ and the dimension of the input $(d+1)$ are both assumed to be $\geq 1$ (which otherwise would render the proof useless). Note then that both terms in the equation are continuous and strictly increasing in $r_n$. Further, $2(1+\frac{1}{d+1})=\text{constant}>2 \; \forall \; d < \infty$, define that constant as $c$ for ease of notation. To show that the dominant term is indeed $r_n^c$ we use the following:
	$$
	\lim_{r_n\rightarrow \infty}\frac{\log(r_n)}{r_n^c}=\lim_{r_n\rightarrow \infty}\frac{\frac{1}{r_n}}{cr_n^{c-1}}=\frac{0}{\infty} = 0
	$$
	Where we used l'Hôpitals Rule in the first step. We then claim the following:
	
	\begin{remark}[Behaviour of $r_n$]
		$$
		r_n = \mathcal{O}(\Big ( \frac{n}{\log(n)}\Big )^{\frac{1}{c}})
		$$
	\end{remark}

	To prove the claim we proceed in two steps, first consider the case where:
	$$
	r_n^c \leq \sqrt{n} \; \textnormal{then}: \frac{r_n^c}{\frac{n}{\log(n)}} \leq \frac{\sqrt{n}}{\frac{n}{\log(n)}} = \frac{\log(n)}{\sqrt{n}}
	$$
	By using the the same procedure as in the previous claim twice and the fact that $n \geq 1$, we have $\frac{\log(n)}{\sqrt{n}} \textnormal{ is continuous and }\rightarrow 0, \; \textnormal{as } n \rightarrow \infty$ and hence it is bounded by some constant $const._1$ which results in:
	$$
	r_n \leq \Big ({const._1 \frac{n}{\log(n)}}\Big )^\frac{1}{c}
	$$
	
	In the second case consider $r_n^c > \sqrt{n}$ then we have:
	$$
	\frac{r_n^c\log(n)}{2}=r_n^c\log(\sqrt{n}) < \text{constant}* r_n^c\log(r_n) \textnormal{ which we saw above is set as } = \mathcal{O}(n) 
	$$
	Noting that by the definition of $\mathcal{O}(n)$ this implies: $\frac{r_n^c\log(n)}{2} \leq \text{const.}_2 n$ and hence:
	$$
	r_n \leq \Big (2 \text{const.}_2 \frac{n}{\log(n)} \Big)^{\frac{1}{c}}
	$$
	Taking these conditions together yields:
	$$
	r_n \leq \max\{const._1^\frac{1}{c}, (2const._2)^\frac{1}{c} \} \Big (\frac{n}{\log(n)}\Big )^\frac{1}{c} = \mathcal{O}\bigg(\Big ( \frac{n}{\log(n)}\Big )^\frac{1}{c}\bigg)
	$$
	
	Then, because we set $\delta_n = ||\theta_0 - \pi_n\theta_0||$, we can combine the approximation error rate with the above and get the final convergence rate of:
	$$
	||\hat{\theta}_n - \theta_0|| = \mathcal{O}\bigg([n/\log(n)]^{\frac{-(1+ \frac{2}{d+1})}{4(1+\frac{1}{1+d})}}\bigg)
	$$
	
	Which completes the proof.
	
\end{proof}

\subsection{Metrics}
Where not else defined, the following metrics apply:
\vspace{3mm}

\noindent \textbf{Mean Squared Error} (MSE)
$$
\text{MSE}=\frac{1}{T}\sum_{t=1}^{T}(y_t - \hat{y_t})^2
$$

Note that when we run MC simulations, we observe the true function and the MSE is calculated on the true $y_t$. For applications we also define the MSPE, where we only train the model with data up to $T$ (ie. observations $T+1,...,T'$ are not used in the modelling phase)
\vspace{3mm}

\noindent \textbf{Mean Squared Prediction Error }(MSPE)
$$
\text{MSPE}=\frac{1}{(T'-T)}\sum_{t=(T+1)}^{T'}(y_{t} - \hat{y}_{t})^2
$$
\vspace{3mm}

\noindent \textbf{Squared Bias} ($\textnormal{Bias}^2$)
$$
\textnormal{Bias}^2=\big(\frac{1}{N}\sum_{i=1}^{n}(\hat{y}_{t,i})-y_t\big)^2,\;\;\; t=1,...,T
$$
Where $N$ is the number of simulations and $T$ number of observations.
\vspace{3mm}

\noindent \textbf{Variance of the estimator} (Variance$_e$)
$$
\textnormal{Var}_e=\frac{1}{N}\sum_{j=1}^{N}\big(\frac{1}{N}\sum_{i=1}^{n}(\hat{y}_{t,i})-\hat{y}_{t,j}\big)^2,\;\;\; t=1,...,T
$$

\subsection{Chaos Model}
To demonstrate the universal approximation property of the neural network, and demonstrate the ineptitude of the linear model the following process which is an adapted version of the one used in \cite{kuan_white_1994} was used to generate the datapoints:

\begin{equation}\label{apB:model_chaos}
	y_{t} = 0.3 y_{t-1} + \frac{22}{\pi}sin(2\pi y_{t-1} + 0.\overline{33})
\end{equation}

Below is a table that compares the performance on a quantitative level. Note that we did not yet need to care about overfitting as the function that was generated did not contain any noise but only the "true" data generating process.

The ARIMA model chose the correct linear specification with an AR(1). The Neural Network was implemented in \textit{Keras} (\cite{chollet2015keras}) with 50 hidden units.

\subsection{Irregular IID Model}

This model contains a highly nonlinear relationship with between the $x$ and $y$ variables. Additionally, white noise has been added to the process to investigate the bias-variance tradeoff. The model is generated according to:

\begin{equation}\label{apB:iid_model}
	y_i = 0.4(x_i-10)^3 + 0.1(\frac{x_i}{7})^7 + 600\sin(2x_i) + \mathbbm{1}_{x_i > 1}(-800\sin(2x_i) - 200) + \varepsilon_i
\end{equation}
$$
i = 1,...,400 \;\;\; x_i \sim \textnormal{Unif}([-10,10]), \;\;\; \varepsilon_i \sim \mathcal{N}(0, 16^2)
$$


\subsection{High Dimensional IID Model - 1}

The high dimensional model 1 is generated according to the process:
\begin{equation}\label{apB:high_dim_1}
	y_i = \sum_{i=1}^{k}f_i(X_i) + \varepsilon_i
\end{equation}
$$
k={2,5,10,15}, \;\;\; X_i=
\begin{cases}
	\sim \textnormal{Unif}([0,3]), & \text{relevant} \\
	\sim \textnormal{Unif}([-1,1]), & \text{otherwise}
\end{cases}, \;\;\;
\varepsilon_i \sim \mathcal{N}(0, 9^2)
$$

\begin{table}[!htbp] \centering
	\begin{tabular}{@{\extracolsep{5pt}} c|c} 
		\hline \\[-1.8ex] 
		& Function \\ 
		\\[-1.8ex]\hline 
		\hline \\[-1.8ex] 
		$f_1$ & $3.5\sin(x)$ \\ 
		$f_2$ & $8\log(\max(|x|,1))$ \\ 
		$f_3$ & $2x^4$ \\ 
		$f_4$ & $-0.4(x^2 + x^3 + 0.1\log(\max(|x|,0.5)))$ \\ 
		$f_5$ & $-4x^3$ \\ 
		$f_6$ & $7x^2$ \\ 
		$f_7$ & $2\log(\max(|x|, 0.3))^3$ \\ 
		$f_8$ & $|x|$ \\ 
		$f_9$ & $ -(0.9x^2 + x^3)/(\max(\sin(x)+2x^5,0.9))$ \\ 
		$f_{10}$ & $-4\cos(x)$ \\ 
		$f_{n > 10}$ & $x$ \\ \cline{1-2}
	\end{tabular} 
	\caption[Table of functions in high dimensional setting]{The functions generating the high dimensional model. If the number of relevant predictors exceeds $10$, the functions the subsequent predictors are just linear additive terms. Note that \textit{all} generated models contain the number of nuisance parameters indicated in addition to the gaussian error term.} 
	\label{apB:functions_high_dim} 
\end{table} 

Estimation is then based on all relevant variables plus the specified number of irrelevant components. For an overview of the used functions in the high dimensional model, refer to table \ref{apB:functions_high_dim}.

\subsection{Portfolio Construction and Testing}\label{apB:application_specification}

The Portfolio consists of three stocks and three other assets which are summarised in table \ref{table:appendix_stocks} (for their log-returns). Each asset is then weighted by a random weight under full investment constraint and no short selling (ie. for for every weight $w_i$ we have $w_i\geq 0$ and $\sum_{i=1}^{6}w_i=1$) to construct a portfolio. 

We estimate an MGARCH(1,1) model and an SANN with a CaViaR specification. A GARCH-type CaViaR model was not estimated due to the model structure of the \texttt{rmgarch} package. As before, we specify a DCC-model with a multivariate normal distribution. The SANN, contains 80 hidden ReLU units in the first layer and 5 ReLU units in the second. We allow for up to two lags of the squared returns of variables from the portfolio. The hyperparameter for the SANN $l_1$-penalization was only optimized once and then kept constant to save compute-time.

We then draw 50 sets of random weights for the portfolio and calculate the $VaR_\alpha$ for $\alpha=\{0.01, 0.05\}$. We then run a test based on \cite{christoffersen_comparingVaR} to test whether the proportion of exceedances in the forecasting period is significantly different from the expected proportion, where the null H$_0$ is that the exceedances are in accordance with the expected number. We also calculate a second test based on \cite{christoffersen_durationVaR}, to assess whether the violations are independent of each other. The null H$_0$ in this case is that the violations are independent. If we reject the null in such a scenario we have a strong indication that the model is not able to capture all the dynamics present in the data. Both tests are implemented in the \texttt{rugarch} package for R. 

Figure \ref{fig:traintestsplitportfolio}, depicts the train test split and the log-returns of a portfolio with equal weighting ($w_i=w_1=1/6$) of the assets. The beginning of the test data still has a higher volatility from the great recession and the European debt crisis in late 2011 (with the explosion of the long-term interest rate for Greece) is also visible. 

\begin{figure}
	\centering
	\includegraphics[width=0.7\linewidth]{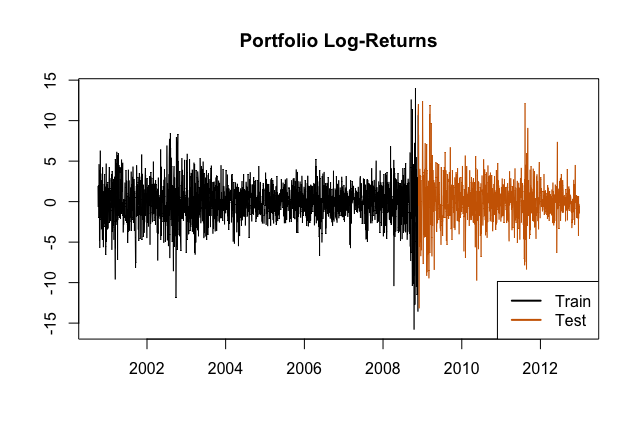}
	\caption[Train-Test split for portfolio VaR estimation]{Train-Test split for portfolio VaR estimation. The portfolio consists of all equally weighted assets. }
	\label{fig:traintestsplitportfolio}
\end{figure}

\begin{table}[!htbp] \centering 
	\begin{tabular}{@{\extracolsep{5pt}}lccccccc} 
		\\[-1.8ex]\hline 
		\hline \\[-1.8ex] 
		Statistic & \multicolumn{1}{c}{N} & \multicolumn{1}{c}{Mean} & \multicolumn{1}{c}{St. Dev.} & \multicolumn{1}{c}{Min} & \multicolumn{1}{c}{Pctl(25)} & \multicolumn{1}{c}{Pctl(75)} & \multicolumn{1}{c}{Max} \\ 
		\hline \\[-1.8ex] 
		GE & 2,992 & $-$0.032 & 1.312 & $-$8.434 & $-$0.566 & 0.548 & 11.052 \\ 
		Walmart & 2,992 & $-$0.017 & 0.435 & $-$2.483 & $-$0.242 & 0.195 & 3.056 \\ 
		Altria & 2,992 & $-$0.010 & 0.659 & $-$6.315 & $-$0.298 & 0.306 & 4.126 \\ 
		Brent & 2,992 & $-$0.002 & 0.993 & $-$8.485 & $-$0.543 & 0.562 & 7.696 \\ 
		Gold & 2,992 & $-$0.0001 & 1.006 & $-$7.177 & $-$0.492 & 0.538 & 7.906 \\ 
		USD/CHF & 2,992 & 0.004 & 1.003 & $-$6.928 & $-$0.565 & 0.564 & 12.452 \\ 
		\hline \\[-1.8ex] 
	\end{tabular} 
	\caption[Summary statistics of portfolio assets]{Summary statistics for the used assets over the time period 01.01.2000-31.12.2012. \label{table:appendix_stocks}}
\end{table} 

\begin{table}[!htbp] \centering 
	\begin{tabular}{@{\extracolsep{5pt}} ccccccc} 
		\\[-1.8ex]\hline 
		\hline \\[-1.8ex] 
		& GE & Walmart & Altria & Brent & Gold & USD\_CHF \\ 
		\hline \\[-1.8ex] 
		GE & $1.722$ & $0.242$ & $0.238$ & $0.101$ & $$-$0.021$ & $0.033$ \\ 
		Walmart & $0.242$ & $0.190$ & $0.065$ & $$-$0.020$ & $$-$0.021$ & $0.034$ \\ 
		Altria & $0.238$ & $0.065$ & $0.435$ & $0.022$ & $$-$0.015$ & $0.027$ \\ 
		Brent & $0.101$ & $$-$0.020$ & $0.022$ & $0.987$ & $0.145$ & $$-$0.177$ \\ 
		Gold & $$-$0.021$ & $$-$0.021$ & $$-$0.015$ & $0.145$ & $1.012$ & $$-$0.215$ \\ 
		USD/CHF & $0.033$ & $0.034$ & $0.027$ & $$-$0.177$ & $$-$0.215$ & $1.006$ \\ 
		\hline \\[-1.8ex] 
	\end{tabular} 
	\caption[Unconditional variance-covariance matrix of assets]{Unconditional variance-covariance matrix for the assets from the portfolio. \label{table:unconditonal_covariance_matrix} } 
\end{table} 

\begin{figure}
	\centering
	\includegraphics[width=0.7\linewidth]{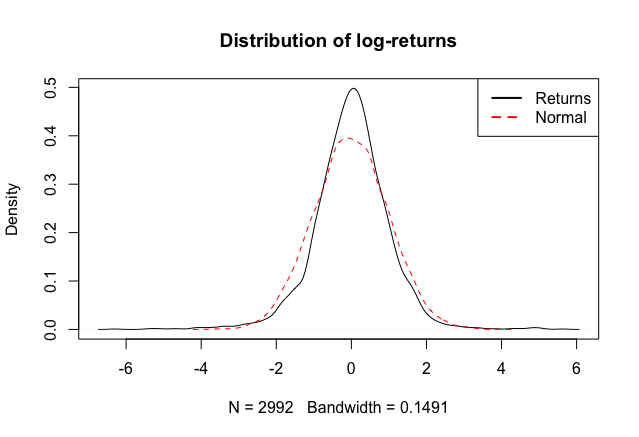}
	\caption[Return distribution of portfolio]{Return distribution of (scaled) log returns with a reference $\mathcal{N}(0,1)$ distribution. Clearly visible is the excess kurtosis of the returns.}
	\label{fig:distributionofreturns}
\end{figure}

\clearpage

\end{document}